\documentclass[DIV=classic,a4paper,10pt]{myart}
\usepackage{todonotes}
\usepackage{caption} 
\usepackage{subcaption}
\usepackage{multirow}
\usepackage{multicol}
\usepackage{rotating}

\usepackage[mathscr]{euscript}
\KOMAoptions{DIV=last}

\usepackage{caption}
\usepackage{subcaption}

\newcommand{\norm}[1]{\left\Vert#1\right\Vert}
\newcommand{\abs}[1]{\left\vert#1\right\vert}

\begin{document}
 
\title{Multivariate Shortfall Risk Allocation and Systemic Risk}

\author[a,1,t2]{Yannick Armenti}
\author[a,2,t1,t4]{St\'ephane Cr\'epey}
\author[b,3,t5]{Samuel Drapeau}
\author[c,4,t3,t4]{Antonis Papapantoleon}

\address[a]{Universit\'e d'Evry, 23 Boulevard de France, 91037 Evry, France}
\address[b]{School of Mathematical Sciences \& Shanghai Advanced Institute for Finance (CAFR/CMAR), Shanghai Jiao Tong University, 211 West Huaihai Road, Shanghai, P.R. 200030 China}
\address[c]{Institute of Mathematics, Technical University Berlin, Stra\ss e des 17. Juni 136, 10623 Berlin, Germany}

\eMail[1]{yannick.armenti@gmail.com}
\eMail[2]{stephane.crepey@univ-evry.fr}
\eMail[3]{sdrapeau@saif.sjtu.edu.cn}
\eMail[4]{papapan@math.tu-berlin.de}

\myThanks[t1]{Financial support from the EIF grant ``Collateral management in centrally cleared trading'', from the Chair ``Markets in Transition'', F\'ed\'eration Bancaire Fran\c caise, and from the ANR 11-LABX-0019.}
\myThanks[t2]{Financial support from LCH.Clearnet Paris.}
\myThanks[t3]{Financial support from the EIF grant ``Post-crisis models for interest rate markets''.}
\myThanks[t4]{Financial support from the DAAD PROCOPE project ``Financial markets in transition: mathematical models and challenges''.}
\myThanks[t5]{Financial support from the National Science Foundation of China, ``Research Fund for International Young Scientists'', Grant number 11550110184.}

\abstract{
The ongoing concern about systemic risk since the outburst of the global financial crisis has highlighted the need for risk measures at the level of sets of interconnected financial components, such as portfolios, institutions or members of clearing houses.
The two main issues in systemic risk measurement are the computation of an overall reserve level and its allocation to the different components according to their systemic relevance.
We develop here a pragmatic approach to systemic risk measurement and allocation based on multivariate shortfall risk measures, where acceptable allocations are first computed and then aggregated so as to minimize costs.
We analyze the sensitivity of the risk allocations to various factors and highlight its relevance as an indicator of systemic risk.
In particular, we study the interplay between the loss function and the dependence structure of the components.
Moreover, we address the computational aspects of risk allocation.
Finally, we apply this methodology to the allocation of the default fund of a CCP on real data.
}
\keyWords{Systemic risk, risk allocation, multivariate shortfall risk, sensitivities, numerical methods, CCP, Default Fund.}


\date{\today}
\ArXiV{1507.05351}
\keyAMSClassification{91G, 91B30, 91G60}
\maketitle\frenchspacing

\section{Introduction}\label{s:intr}

The ongoing concern about systemic risk since the onset of the global financial crisis has prompted intensive research on the design and properties of multivariate risk measures.
In this paper, we study the risk assessment for financial systems with interconnected risky components, focusing on two major aspects, namely:
\begin{itemize}
    \item The quantification of a monetary risk measure corresponding to an overall reserve of liquidity such that the whole system can overcome unexpected stress or default scenarios;
    \item The allocation of this overall amount between the different risk components in a way that reflects the systemic risk of each one.
\end{itemize}
Our goal is fourfold.
First, we introduce a theoretically sound but numerically tractable class of systemic risk measures.
Second, we study the impact of the intrinsic dependence on the risk allocation and it sensitivity.
Third, we address the computational aspect and challenges of systemic risk allocation.
Finally, we present empirical results, based on real data provided by LCH S.A., on the risk allocation of the default fund of a CCP.

\paragraph{Review of the Literature:}
Monetary risk measures have been the subject of intensive research since the seminal paper of \citet{artzner1999}, which was further extended by \citet{foellmer2002} and \citet{fritelli2002}, among others. 
The corresponding risk measures, including conditional value-at-risk by \citet{artzner1999}, shortfall risk measures by \citet{foellmer2002} or optimized certainty equivalents by \citet{bental2007}, can be applied in a multivariate framework that models the dependence of several financial risk components.
Multivariate market data-based risk measures include the marginal expected shortfall of \citet{AcharyaPedersenPhilipponRichardson10}, law invariant convex risk measures for portfolio vectors of \citet{rueschendorf2006}, the systemic risk measure of \citet{AcharyaEngleRichardson12} and \citet{BrownleesEngle12}, the delta conditional value-at-risk of \citet{AdrianBrunnermeier11} or the contagion index of \citet{Cont2013}.
In parallel, theoretical economical and mathematical considerations have led to multivalued and set-valued risk measures, in static or even dynamic setup; see for instance \citet{CascosMolchanov14,hamel2011} and \citet{JouiniMeddebTouzi04}. 

Recently, the risk management of financial institutions raised concerns about the allocation of the overall risk among the different components of a financial system. 
A bank, for instance, for real time monitoring purposes, wants to channel to each trading desk a cost reflecting its responsibility in the overall capital requirement of the bank. 
A central clearing counterparty --- CCP for short, also known as a clearing house --- is interested in quantifying the size of the so-called default fund and allocating it in a meaningful way among the different clearing members, see \citep{cont2015,crepey2016,ghamami2016}.
On a macroeconomic level, regulators are considering to require from financial institutions an amount of capital reflecting their systemic relevance.
The aforementioned approaches can only address the allocation problem indirectly, through the sensitivity of the risk measure with respect to the different risk components.
For instance, the so-called \emph{Euler rule} allocates the total amount of risk according to the marginal impact of each risk factor.
However, a practical limitation of the Euler rule is that it is based on G\^ateaux derivatives which in general is difficult to compute beyond simple cases.
Also this Euler rule consider the marginal risk of one element with respect to the full system rather than the marginal risk with respect to each individual components.
In addition, the Euler risk allocation does not add up to the total risk, unless the univariate risk measure that is used in the first place is sub-additive, see \citep{tasche2008}.
In other words, the Euler rule does not automatically fulfill the so-called \emph{full allocation} property.
The work by \citet{cheridito2014} addresses systematically the question of allocation of systemic risk with regard to certain economic properties:
\begin{itemize}
    \item Full allocation: the sum of the components of the risk allocation is equal to the overall risk measure;
    \item Riskless allocation: if a risk factor is riskless, the corresponding component of the risk allocation is equal to it;
    \item Causal responsibility: any system component bears the entire additional costs of any additional risk that it takes.
\end{itemize} 
More specifically, \citet{cheridito2014} propose a framework where an overall capital requirement is first determined by utility indifference principles and then allocated according to a rule such that the above three properties are fulfilled, at least at a first order level of approximation.
In fact, as far as dependence is concerned, whether the last two properties should hold is debatable.
One may argue that each component in the system is not only responsible for its own risk taking but also for its relative exposure to other components.
This is also what comes out from the present study, see Section \ref{ss:sensi}.
In a general framework, \citet{kromer2016} characterized systemic risk out of axioms allowing for a decomposition between and aggregation function and a univariate risk measure.
In the spirit of this aggregation function, in two recent papers, \citet{feinstein2015} and \citet{BiaginiFouqueFrittelliMeyer15} proposed a general approach similar in spirit to ours.
We precise thereafter and later in the paper the relationship to these references and in which sense our take on differs.

\paragraph{Contribution and Outline of the Paper:}
Our approach addresses simultaneously the design of an overall risk measure regarding a financial system of interconnected components and the allocation of this risk measure among the different risk components; the emphasis lies on the allocation and its sensitivities.
In contrast to \citep{cheridito2014, chen2013}, we \emph{first allocate} the monetary risk among the different risk components and \emph{then aggregate} and minimizes the risk allocations in order to obtain the overall capital requirement.
As previously mentioned, \citep{kromer2016}, \citep{feinstein2015} and \citep{BiaginiFouqueFrittelliMeyer15} develop approaches in a similar spirit, covering allocation first followed by aggregation, in general frameworks with different aggregation procedures.
They focus on the resulting risk measure, conducting systematic studies of their properties in terms of set valued functions, diversification and monotonicity, among others.
The multivariate shortfall risk measure of this paper can be viewed as a special case of their definition, in a way precised in Remark \ref{rem:WebBia02}.
Sharing with these references the ``allocate first, then aggregate'' perspective, our approach is restricted to a systemic extension of shortfall risk measures, see \citep{foellmer2002}, based on multivariate loss functions.
However, in contrast to the aforementioned references, we focus on the resulting risk allocation in terms of existence, uniqueness, sensitivities and numerical applications.
In our framework, the \emph{systemic risk} is the risk that stems specifically from the intrinsic dependence structure of an interconnected system of risk components.
In this perspective, the risk allocation and its properties provide a ``cartography'' of the systemic risk, see Section \ref{sec04:computation} on the numerical aspects of risk allocation and the empirical study in Section \ref{sec05:empiricalstudy} on real data for an illustration thereof.
It turns out that special care has to be given to the specifications of the loss function in order to stress the systemic risk.
In \citep{BiaginiFouqueFrittelliMeyer15}, by allowing random allocations, the impact of the interdependence structure can be observed in the future.
Such random allocations may be interesting in view of a posterior management of defaults.
By contrast, our deterministic allocation is sensitive to the dependence of the system already at the moment of the quantification, see Section \ref{sec03:sensis} and see \emph{a contrario} Proposition \ref{prop:marginals}.
We study the sensitivity of the risk allocation with respect to external shocks as well as internal dependence structure.
We show in particular that a causal responsibility can be derived in marginal terms, see Proposition \ref{prop:causalresponsability}.
In addition, we discuss computational aspects of risk allocation and finally, we provide an empirical study on the risk allocation of a default fund of a CCP based on real data provided by LCH S.A.

The univariate shortfall risk measure as a law invariant risk measure holds additional properties as an operator on probability distributions.
Indeed, as studied by \citet{weber2006} and \citet{schied2014}, it has some continuity properties with respect to the $\psi$-weak topology on distributions.
It has been furthermore characterised in \citep{weber2006} as the only convex law invariant convex risk measure on the level of distributions and therefore the unique one having elicitability properties, a wishful statistical property, see \citep{osband1985,bellini2015}. 
Extensions of these results, such as elicitability characterization in multidimensional case as proposed by \citet{ziegel2014} and \citet{ziegel2015}, as well as the axiomatic characterization along the lines of \citep{weber2006}, are highly non trivial and therefore let for further study.
A set-valued multivariate shortfall risk measure has been introduced by \citet{ararat2014}.
However, allocation is the not focus of their work and the loss function that they then consider is decoupled in the sense of \ref{C2}, which from our viewpoint is too restrictive in view of Proposition \ref{prop:marginals}.

The paper is organized as follows:
Section \ref{sec01:msra} introduces the class of systemic loss functions, acceptance sets and risk measures that we use in the paper.
Section \ref{sec02:riskallocation} establishes the existence and uniqueness of a risk allocation.
Section \ref{sec03:sensis} focuses on sensitivities with respect to external shocks, dependence structure, nature of the loss function as well as the properties of full allocation, causal responsibility and riskless allocation mentioned beforehand.
Section \ref{sec04:computation} discusses the computational aspects and challenges of risk allocation.
Section \ref{sec05:empiricalstudy}, applies our approach to the concrete allocation of the default fund of a CCP.
Appendices \ref{appendix:01} and \ref{appendix:02} gather classical facts from convex optimization and results on multivariate Orlicz spaces.
Appendix \ref{appendix:03} provides additional insight on the data of the empirical study.


\subsection{Basic Notation}
Let $x_k$ denote the generic coordinate of a vector $x \in \mathbb{R}^d$, and $e_k$ the $k$-th unit vector.
By $\geqslant$ we denote the lattice order on $\mathbb{R}^d$, that is, $x \geqslant y$ if and only if $x_k\geq y_k$ for every $1\le k\le d$.
We denote by $\norm{\cdot}$ the Euclidean norm and by ${}^\pm,\wedge, \vee, \abs{\cdot}$ the lattice operations on $\mathbb{R}^d$.
For $x,y \in \mathbb{R}^d,$ we write $x>y$ for $x_k>y_k$ componentwise, $x\cdot y = \sum x_ky_k$, $xy=(x_1y_1,\ldots ,x_dy_d)$ and $x/y=(x_1/y_1,\ldots, x_d/y_d)$.
We denote by $f^\ast(y)=\sup_{x}\{x\cdot y-f(x)\}$ the convex conjugate of a function $f:\mathbb{R}^d\to [-\infty,\infty]$, and for $C\subseteq \mathbb{R}^d$, we denote by $\delta(\cdot |C)$ the indicator function of $C$ being equal to $0$ on $C$ and $\infty$ otherwise.

Let $(\Omega,\mathcal{F},P)$ be a probability space, and denote by $L^0(\mathbb{R}^d)$ the space of $\mathcal{F}$-measurable $d$-variate random variables on this space identified in the $P$-almost sure sense.
The space $L^0(\mathbb{R}^d)$ inherits the lattice structure of $\mathbb{R}^d$, hence we can use the above notation in a $P$-almost sure sense.
For instance, for $X$ and $Y$ in $L^0(\mathbb{R}^d),$ we say that $X\geqslant Y$ or $X>Y$ if $P[X\geqslant Y]=1$ or $P[X>Y]=1$, respectively.
Since we mainly deal with multivariate functions or random variables, to simplify notation we drop the reference to $\mathbb{R}^d$ in $L^0(\mathbb{R}^d)$, writing simply $L^0$ unless necessary.


\section{Multivariate Shortfall Risk}\label{sec01:msra}

Let $X=(X_1,\ldots,X_d) \in L^0$ be a random vector of financial losses, that is, negative values of $X_k$ represent actual profits.
We want to determine an overall monetary measure $R(X)$ of the risk of $X$ as well as a sound risk allocation $RA_k(X), k=1,\ldots,d,$ of $R(X)$ among the $d$ risk components. 
We consider a flexible class of risk measures defined by means of loss functions and sets of acceptable monetary allocations.
This class allows us to discuss in detail the properties of the resulting risk allocation as an indicator of systemic risk.
Inspired by the shortfall risk measure introduced in \citep{foellmer2002} in the univariate case, we start with a loss function 
$\ell$ defined on $\mathbb{R}^d,$ used to measure the expected loss $E[\ell(X)]$ of the financial loss vector $X$.
\begin{definition}\label{def:loss}
    A function $\ell:\mathbb{R}^d\to (-\infty, \infty]$ is called a \emph{loss function} if
    \begin{enumerate}[label=\textbf{(A\arabic*)}]
        \item\label{a1} $\ell$ is increasing, that is, $\ell(x)\geq \ell(y)$ if $x\geqslant y$;
        \item\label{a2} $\ell$ is convex, lower semi-continuous with $\inf \ell <0$;
        \item\label{a3} $\ell(x)\geq \sum x_k-c$ for some constant $c$.
    \end{enumerate}
    A loss function $\ell$ is \emph{permutation invariant} if $\ell(x)=\ell(\pi(x))$ for every permutation $\pi$ of the components.
\end{definition}
A risk neutral assessment of the losses corresponds to $E[\sum X_k]=\sum E[ X_k]$.
Thus, \ref{a3} expresses a form of risk aversion, whereby the loss function puts more weight on high losses than a risk neutral evaluation.
As for \ref{a1} and \ref{a2}, they express the respective normative facts about risk that ``the more losses, the riskier'' and ``diversification should not increase risk''; see \citep{drapeau2013} for related discussions.
\begin{remark}\label{rem:permutation}
    The choice of the terminology ``loss function'' stems from \cite{foellmer2002} for which this paper is a multivariate extension.
    Our notion of a loss function coincide with the one of ``aggregation function'' in \citep{feinstein2015, BiaginiFouqueFrittelliMeyer15}, in the sense that it aggregate several loss profiles into a univariate random variable for which it can be decided whether or not it is acceptable, see Remark \ref{rem:WebBia02}.
    Due to the obvious extension from the shortfall risk measure, throughout this paper we stick to the terminology ``loss function''.

    As for the permutation invariance, the considered risk components are often of the same type --- banks, members of a clearing house or trading desks within a trading floor.
    In that case, the loss function should not discriminate a particular component against another.
\end{remark}
\begin{example}\label{ex:loss-functions}
    Let $h:\mathbb{R}\to \mathbb{R}$ be a one-dimensional loss function such as for instance
    \begin{equation*}
        h(x)=x^+-\beta x^-,\, 0\leq \beta <1,\quad h(x)= x+(x^+)^2/2\quad \text{ or }\quad h(x)=e^x-1.
    \end{equation*}
    Using these as building blocks, we obtain the following classes of multivariate loss functions, which will be used for illustrative purposes in the discussion of systemic risk, see Section \ref{sec02:riskallocation} and \ref{sec03:sensis}.
    \begin{enumerate}[label=\textbf{(C\arabic*})]
        \item\label{C1} $\ell(x)=h(\sum x_k)$;
        \item\label{C2} $\ell(x)=\sum h(x_k)$;
        \item\label{C3} $\ell(x)=\alpha h(\sum x_k)+\beta\sum h(x_k)$, where $\alpha,\beta \geq 0$ non both zero.
    \end{enumerate}
    Note that each of these loss functions are permutation invariant.
\end{example}
For integrability reasons we consider loss vectors in the following multivariate Orlicz heart:\footnote{Orlicz spaces are natural spaces in this context. The theory of Orlicz spaces has been used for long in the theory of risk measures, see \citep{delbaen2002,biagini2008,cheridito2009,biagini2010}.}
\begin{equation*}
    M^\theta =\left\{X \in L^0 \colon E\left[ \theta\left( \lambda X \right) \right]<\infty\text{ for all }\lambda \in \mathbb{R}_+\right\},
\end{equation*}
where $\theta(x)=\ell(\abs{x})$, $x \in \mathbb{R}^d$; see Appendix \ref{appendix:02}.
\begin{remark}
\end{remark}
\begin{definition}\label{e:dA}
    A monetary allocation $m \in \mathbb{R}^d$ is \emph{acceptable} for $X$ if
    \begin{equation*}
        E\left[\ell\left( X-m \right) \right]\leq 0.  
    \end{equation*}
    We denote by 
    \begin{equation}\label{eq:acceptance}
        A(X):=\left\{m \in \mathbb{R}^d: E\left[\ell\left( X-m \right) \right]\leq 0\right\}
    \end{equation}
    the corresponding set of \textit{acceptable monetary allocations}.
\end{definition}
\begin{example}\label{e:clearinghouse}
In a centrally cleared trading setup, each clearing member $k$ is required to post a default fund contribution $m_k$ in order to make the risk of the clearing house acceptable with respect to a risk measure accounting for extreme and systemic risk.
The default fund is a pooled resource of the clearing house, in the sense that the default fund contribution of a given member can be used by the clearing house not only in case the liquidation of this member requires it, but also in case the liquidation of another member requires it.
For the determination of the default fund contributions, the methodology of this paper can be applied to the vector $X$ defined as the vector of stressed losses-and-profits of the clearing members.
According to the findings of Section \ref{sec02:riskallocation} and \ref{sec03:sensis}, a ``systemic'' loss function such as \ref{C3} with $\alpha>0$  would be consistent with the purpose of a default fund.
Note however that our setup applied to clearing houses takes the view of a closed system, so an internal assessment.
In principle we ignore additional systemic risk such as a competition between clearing houses with common membership, or the external risk to which these members may be subject to, as addressed for instance in \citep{glasserman2016}.
However, our method could also assess such a systemic risk by taking $X$ as the overall vector of positions of each member in each clearing house.
\end{example}

The next proposition gathers the main properties of the sets of acceptable monetary allocations.
The convexity property in \ref{cond01} means that a diversification between two acceptable monetary allocations remains acceptable.
If a monetary allocation is acceptable, then any greater amount of money should also be acceptable, which is the monotonicity property in \ref{cond01}.
As for \ref{cond02}, it says that, if the losses $X$ are less than $Y$ almost surely, then any monetary allocation that is acceptable for $Y$ is also for $X$.
Next, \ref{cond03} means that a convex combination of allocations acceptable in two markets is still acceptable in the diversified market.
In particular, the acceptability concept pushes towards a greater diversification among the different risk components.
From the viewpoint of a clearing house for instance, a diversified position of its members is preferable to a concentrated one and therefore may enforce default fund allocations that incite its members towards this goal.
Also, from a trading floor supervision, an overall diversified position of the traders is preferable, an incentive which is a current practice, see example \ref{ss:trivariate}.
Finally, \ref{cond04} means that acceptable positions translate with cash in the sense of scalar monetary risk measures \`a la \citep{artzner1999,foellmer2002,fritelli2002}.
As an immediate consequence of these properties, $X\mapsto A(X)$ defines a monetary set-valued risk measure in the sense of \citep{hamel2011}, that is, a set-valued map $A$ from $M^\theta  $ into the set of monotone, closed and convex subsets of $\mathbb{R}^d$.

\begin{proposition}\label{prop:01}
    For $X,Y$ in $M^\theta$, it holds:
    \begin{enumerate}[label=\textit{(\roman*)}]
        \item\label{cond01} $A(X)$ is convex, monotone and closed;
        \item\label{cond02} $A(X)\supseteq A(Y)$ whenever $X\leqslant Y$;
        \item\label{cond03} $A(\alpha X+(1-\alpha)Y)\supseteq \alpha A(X)+(1-\alpha)A(Y),$ for any $\alpha \in (0,1)$;
        \item\label{cond04} $A(X+m)=A(X)+m,$ for any $m\in \mathbb{R}^d$;
        \item\label{cond05} $\emptyset \neq  A(X)\neq  \mathbb{R}^d$.
    \end{enumerate}
    If furthermore 
    \begin{enumerate}[label=\textit{(\roman*)},resume]
        \item\label{cond06} $\ell$ is positive homogeneous, then $A(\lambda X)=\lambda A(X)$ for every $\lambda>0$;
        \item\label{cond07} $\ell$ is permutation invariant, then $A(\pi(X))=\pi(A(X))$ for every permutation $\pi$;
    \end{enumerate}
\end{proposition}
\begin{proof}
    Since $\ell$ is convex, increasing and lower semi-continuous, it follows that $(m,X)\mapsto E[\ell (X-m)]$ is convex and lower semi-continuous, decreasing in $m$ and increasing in $X$.
    This implies the properties \ref{cond01} through \ref{cond03} by Definition \ref{e:dA} of $A(X)$.
    Regarding \ref{cond04}, a change of variables yields
    \begin{align*}
        A(X+m)&=\left\{n \in \mathbb{R}^d\colon E\left[\ell\left( X+m-n \right) \leq 0\right]\right\}=\left\{n+m \in \mathbb{R}^d\colon E\left[\ell\left( X-n \right) \right]\leq 0\right\}=A(X)+m.
    \end{align*}
    As for \ref{cond05}, on the one hand, $\ell(X-m)\searrow \ell(-\infty)< 0$ as $m\to \infty$ component-wise.
    Since $X\in M^\theta$ it follows that $\ell(X) \in L^1$, thus monotone convergence yields $E[\ell(X-m)]\searrow \ell(-\infty)<0$ and in turns the existence of $m \in \mathbb{R}^d$ such that $E[\ell(X-m)]\leq 0$, showing that $A(X)\neq \emptyset$.
    On the other hand, $\ell$ being increasing and such that $\ell(x)\geq \sum x_k-c$, it implies that $\ell(X-m)\geq \sum X_k-\sum m_k -c\nearrow \infty $ as $m\to -\infty$, component-wise.
    Hence, monotone convergence yields $E[\ell(X-m)]\nearrow \infty>0$, therefore there exists $m \in \mathbb{R}^d$ such that $E[\ell(X-m)]>0$, that is, $m \not \in A(X)$.
    As for \ref{cond06}, if $\ell$ is positive homogeneous, for any $\lambda >0$ it holds $E[\ell(\lambda X-m)]=\lambda E[\ell(X-m/\lambda)]$.
    Hence $m$ is in $A(\lambda X)$ if and only if $m/\lambda$ is in $A(X)$ if and only if $m$ is in $\lambda A(X)$.
    Finally, if $\ell$ is permutation invariant, for any permutation $\pi$ it holds $E[\ell(\pi(X)-m)]=E[\ell(\pi(X-\pi^{-1}(m))]=E[\ell(X-\pi^{-1}(m))]$.
    Hence $m$ is in $A(\pi(X))$ if and only if $\pi^{-1}(m)$ is in $A(X)$, if and only if $m$ is in $\pi(A(X))$ showing \ref{cond07}.
\end{proof}

Figure \ref{fig:021} shows sets of acceptable monetary allocations for a bivariate normal distribution with varying correlation coefficient.
The location and shape of these sets change with the correlation: the higher the correlation, the more costly the acceptable monetary allocations, as expected in terms of systemic risk.
As discussed in Sections \ref{sec02:riskallocation} and \ref{sec03:sensis}, this feature is not always immediate and depends on the specification of the loss function.
\begin{figure}[ht]
    \centering
        \includegraphics[width=0.8\textwidth]{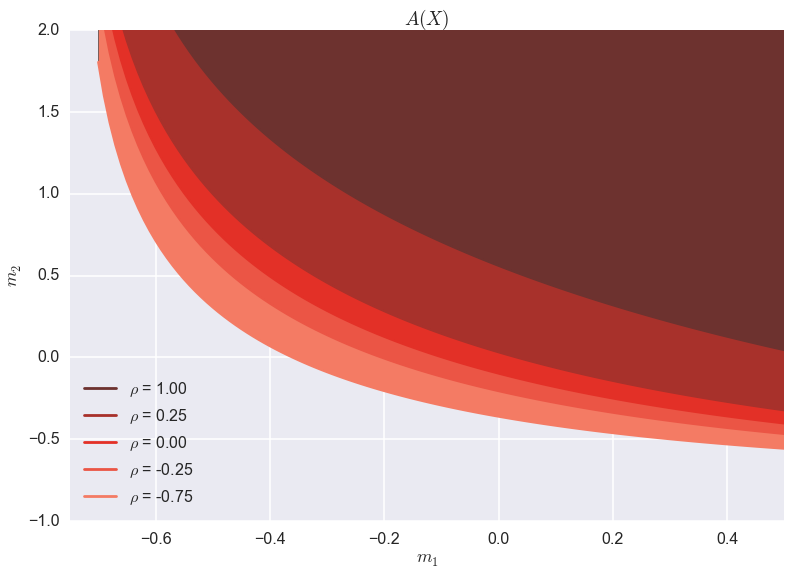}
        \caption{Acceptance sets $A(X)$ corresponding to the case study of Section \ref{subsec:exep01} for different correlations.}
        \label{fig:021}
\end{figure}

Given an acceptable monetary allocation $m\in A(X)$, its aggregated liquidity cost is $\sum m_k$.
The smaller the cost, the better, which motivates the following definition.
\begin{definition}
    The \textit{multivariate shortfall risk} of $X\in M^\theta  $ is  
    \begin{equation} \label{eq:def_shortfall_risk}
        R(X) :=\inf \left\{ \sum m_k\colon m \in A(X)\right\}=\inf \left\{ \sum m_k \colon E\left[ \ell\left( X-m \right) \right] \leq 0\right\}.
    \end{equation}
\end{definition}

\begin{example}
    Following up on the central clearing house Example \ref{e:clearinghouse}, any acceptable allocation $m\in A(X)$ yields a corresponding value for the default fund.
    Clearing houses are in competition with each other, hence they are looking for the cheapest acceptable allocation to require from their members.
\end{example}
\begin{remark}\label{rem:WebBia02}
    When $d=1$, the above definition corresponds exactly to the shortfall risk measure in \citep{foellmer2002}, of which this paper is a  multivariate extension.

    The set valued risk measure $X\mapsto A(X)$ introduced in \eqref{eq:acceptance} can be seen as an example of the set valued systemic risk measures presented in \citep{feinstein2015}, which in their notation translates as follows
    \begin{equation*}
        A(X)=R(Y,k)=\left\{ m \in \mathbb{R}^d\colon Y_{k+m}\in \mathcal{A} \right\}
    \end{equation*}
    where the aggregation is given by $Y_{k+m}=\Lambda(X-k-m)$ for $\Lambda(x)=\ell(x)$ and the acceptance set is $\mathcal{A}:=\{X\colon E[X]\leq 0\}$.
    Their setting considers more general random fields $Y_k$ associated with capital allocations denoted by $k$ accommodating for instance the modelling of financial networks, among others.
    The case we consider can be embedded into \citep[Case (ii), Page 5]{feinstein2015}.
    Even if set valued risk measure is not the primary focus of \citep{BiaginiFouqueFrittelliMeyer15}, it is included in the definition of the acceptance family which, in their notation, is given as follows
    \begin{equation*}
        \mathcal{A}^{m}=\mathcal{A}^{\mathbf{Y}}=\left\{ X\colon E\left[ \ell(X-m) \right] \leq 0\right\}, \quad \mathbf{Y}\in \mathcal{C}
    \end{equation*}
    where $\mathcal{C}=\mathbb{R}^d$ and $\mathbf{Y}=\mathbb{R}^d$.
    The resulting systemic risk measure can also be translated in their notation and denomination in terms of an aggregating function $\Lambda(x)=\ell(x)$, acceptance set $\mathcal{A}=\{X\colon E[X]\leq 0\}$ and a measure of risk $\pi(m)=\sum m_k$, resulting into
    \begin{equation*}
        R(X)=\inf \left\{ \pi(m)\colon \Lambda(X-m)\in \mathcal{A} \right\}.
    \end{equation*}
    Therefore the case we consider can be embedded into the class presented in \citep[Section 1.3]{BiaginiFouqueFrittelliMeyer15}.
\end{remark}
Our next result, which uses the concepts and notation of Appendix \ref{appendix:02}, shows that all the classical properties of the shortfall risk 
measure, including its dual representation, can be extended to 
the multivariate case.
We denote by
\begin{equation*}
    \mathcal{Q}^{\theta^\ast}:=\left\{ \frac{dQ}{dP}:=(Z_1,\ldots,Z_d)\colon Z \in L^{\theta^\ast}, Z\geqslant 0 \text{ and such that }E\left[1\cdot Z\right]=E\left[\sum Z_k\right]=1 \right\}
\end{equation*}
the set of $d$-dimensional measure densities in $L^{\theta^\ast}$ normalized to $E[1\cdot Z]=1$.
For the sake of simplicity, we use the notation $E_Q[X]:=E[dQ/dP \cdot X]$ for $dQ/dP\in \mathcal{Q}^{\theta^\ast}$ and $X\in M^{\theta}$.
\begin{theorem}\label{thm:rep}
    The function
    \begin{equation*}
        R(X)=\inf\left\{ \sum m_k\colon m \in A(X)\right\}, \quad X\in M^\theta  ,
    \end{equation*}
    is real valued, convex, monotone and translation invariant.\footnote{In the sense that $R(X+m)=R(X)+\sum m_k$.}
    In particular, it is continuous and sub-differentiable.
    If $\ell$ is positive homogeneous, then so is $R$.
    Moreover, it admits the dual representation
    \begin{equation}
        R(X)=\max_{Q \in \mathcal{Q}^{\theta^\ast}}\left\{ E_Q\left[  X\right]-\alpha(Q) \right\}, \quad X \in M^\theta  ,
        \label{eq:dualrep}
    \end{equation}
    where the penalty function is given by
    \begin{equation}\label{e:dualrepr}
        \alpha(Q)=\inf_{\lambda>0}  E\left[ \lambda \ell^\ast\left( \frac{dQ}{\lambda dP} \right) \right]  , \quad Q \in \mathcal{Q}^{\theta^\ast}.
    \end{equation}
\end{theorem}
\begin{remark}
    This robust representation can also be inferred from the general results of \citep{FarkasKoch15}.
    However, for the sake of completeness and since the multivariate shortfall risk measure is closely related to a multidimensional version of the optimized certainty equivalent, we give a self contained proof tailored to our context.

    The argumentation follows the original one by \citep{foellmer2002}, which however cannot be directly applied on the product space $\Omega \times \{1,\ldots,d\}$ since the optimization is done here according to multidimensional allocations $m\in \mathbb{R}^d$ rather than one dimensional allocations $m \in \mathbb{R}$.
    Moreover, in the course of our derivation of the dual representation we extend to the multidimensional setting the following relationship between the optimized certainty equivalent and the shortfall risk provided in \citep[Chapter 5.2]{bental2007}
    \begin{equation*}
        R(X)=\inf_{m\in \mathbb{R}}\left\{ m\colon E\left[ \ell(L-m) \right] \leq 0\right\}=\sup_{\lambda>0}\inf_{m\in \mathbb{R}}\left\{ m+\lambda E\left[ \ell(X-m) \right]\right\},
    \end{equation*}
    where
    \begin{equation*}
        S(\lambda,X)=\inf_{m\in \mathbb{R}}\left\{ m+\lambda E\left[ \ell(X-m) \right] \right\}=\sup_{Q\ll P}\left\{ E^Q[X]-E\left[\lambda \ell^\ast\left( \frac{dQ}{\lambda dP} \right)  \right] \right\}
    \end{equation*}
    is the optimized certainty equivalent of $X$.\footnote{Here $\ell$ is a one dimensional loss function and $X$ a one dimensional random variable.}
\end{remark}
\begin{proof}
    By Proposition \ref{prop:01} \ref{cond05}, we have $A(X)\neq \emptyset$ and in turn $R(X)<\infty$.
    If $R(X)=-\infty$ for some $X\in M^\theta  $, then there exists a sequence $(m^n) \subseteq A(X)$ such that $\sum m_k^n \to -\infty$, in contradiction with $0\geq E[\ell(X-m^n)]\geq E[\sum X_k]-\sum m_k^n-c$.
    Hence, $R(X)>-\infty$.
    Monotonicity, convexity and translation invariance readily follow from Proposition \ref{prop:01} \ref{cond02}, \ref{cond03} and \ref{cond04}, respectively.
    In particular, $R$ is a convex, real-valued and increasing functional on the Banach lattice $M^\theta$.
    Hence, by \citep[Theorem 4.1]{cheridito2009}, $R$ is continuous and 
sub-differentiable.
    Therefore, the results recalled in Appendix \ref{appendix:02} and the 
Fenchel-Moreau theorem imply
    \begin{equation}\label{eq:0001}
        R(X)=\sup_{Y \in L^{\theta^\ast}} \left\{ E\left[X \cdot Y\right]-R^\ast(Y) \right\}=\max_{Y \in L^{\theta^\ast}}\left\{ E\left[X \cdot Y\right]-R^\ast(Y) \right\},
    \end{equation}
    where $R^\ast(Y)=\sup\{ E[ X \cdot Y ] -R(X)\colon X \in M^\theta \}$, $Y \in L^{\theta^\ast}$.
    By the bipolar theorem, for $Y\not \geqslant 0$, there exists $K\in M_{\theta}$, $K\geqslant 0$ with $E[Y\cdot K]<-\varepsilon<0$ for some $\varepsilon>0$.
    By monotonicity of $R$, it follows that $R(-\lambda K)\leq R(0)<\infty$ for every $\lambda >0$.
    Hence
    \begin{align*}
        R^\ast(Y)=\sup_{X \in M^\theta}\left\{ E\left[ Y\cdot X \right] -R(X)\right\}\geq \sup_{\lambda >0}\left\{ -\lambda E[Y\cdot K]-R(-\lambda K) \right\}\geq \sup_{\lambda} \lambda \varepsilon -R(0)=\infty,
    \end{align*}
    Furthermore, by translation invariance, setting $X=(r, \ldots,r)$ for $r \in \mathbb{R}$, it follows that 
    \begin{equation*}
        R^\ast(Y)\geq r E\left[ 1 \cdot Y  \right] -R(0)-rd =r \left(E\left[ 1 \cdot Y  \right] -d  \right)-R(0),
    \end{equation*}
    where the right hand side can be made arbitrarily large whenever $E\left[ 1 \cdot Y  \right] \neq d$.
    It shows that the supremum and maximum in \eqref{eq:0001} can be restricted to the set of those $Y \in L^{\theta^\ast}$ such that $Y\geqslant 0$ and $E[1\cdot Y]=1$, that is, can be identified to $\mathcal{Q}^{\theta^\ast}$.
    In order to obtain a more explicit expression of the penalty function $\alpha(Q):=R^\ast(dQ/dP)=R^\ast(Y)$, we set
    \begin{equation*}
        \begin{split}
            L(m,\lambda, X) &= \sum m_k +\lambda E\left[\ell\left( X-m \right) \right] \\
            S (\lambda,X)    &= \inf_{m \in \mathbb{R}^d}L(m,\lambda,X)=\inf_{m\in \mathbb{R}^d}\left\{ \sum m_k +\lambda E\left[\ell\left( X-m \right) \right] \right\}.
        \end{split}
    \end{equation*}
    The functional $X\mapsto S(\lambda, X)$ is a multivariate version of the so called optimized certainty equivalent, see \citep{bental2007}.
    Clearly,
    \begin{equation*}
        R(X)=\inf_{m\in \mathbb{R}^d}\sup_{\lambda>0} L( m,\lambda, X)\geq \sup_{\lambda >0}\inf_{m\in \mathbb{R}^d} L( m,\lambda,  X)=\sup_{\lambda >0}S(\lambda, X).
    \end{equation*}
    Since $A(X)$ is nonempty and monotone, there exists $m \in \text{Int}(A(X))$ and so the Slater condition is fulfilled.
    As a consequence of \citep[Theorem 28.2]{rockafellar1970}, there is no 
duality gap.
    Namely, $R(X) = \sup_{\lambda >0}S(\lambda, X)$.
    Via the first part of the proof, an easy multivariate adaptation of \citep[Chapter 4]{bental2007} and \citep[Chapter 2]{antonis2012} yields
    \begin{equation*}
        S(\lambda, X)=\sup_{Q \in \mathcal{Q}^{\theta^\ast}}\left\{ E_Q\left[ X \right] -E\left[ \left( \ell_{\lambda} \right)^\ast \left( \frac{dQ}{dP} \right)\right] \right\},
    \end{equation*}
    where $\ell_{\lambda}(m)=\lambda\ell(m)$, hence $\ell_\lambda^\ast(m^\ast)=\lambda\ell^\ast(m^\ast /\lambda )$.
    Combining this with $R(X)=\sup_{\lambda>0}S(\lambda,X),$ the dual representation \eqref{e:dualrepr} follows.
\end{proof}
\begin{example}
    We consider the two positive homogeneous loss functions of the empirical study:
    \begin{align}
        \ell_1(x) &=\beta \sum x_k^+-\alpha \sum x_{k}^- \label{eq:loss01}\\
        \ell_2(x) &=\beta \sum x_k^+-\alpha \sum x_{k}^-+\beta\sum_{k<j}\left( x_k+x_j \right)^+-\alpha\sum_{k<j}(x_k+x_j)^- \label{eq:loss02}
    \end{align}
    for $0<\alpha<1<\beta$. 
    A simple computation yields that $\ell_i^\ast=\delta(\cdot | C_i)$ where
    \begin{equation*}
        \begin{split}
            C_1&=\left\{x\colon \alpha \leq x_k\leq \beta \text{ for all }k  \right\}\\
            C_2&=\left\{ x=\sum_{1\leq j\leq d} x_{0j}e_k+\sum_{1\leq k<j\leq d}x_{kj}(e_k+e_j)\colon \alpha\leq  x_{kj}\leq \beta \text{ for all }0\leq k<j\leq d\right\}
        \end{split}
    \end{equation*}
    Note that $[\alpha,\beta]=C_1\subseteq C_2 \subseteq [\alpha,d\beta]$ where $\alpha$ and $\beta$ are identified with their vector of equal components.
    Furthermore, $d\beta$ is an extreme point of $C_2$.
    It follows in particular that $R_1\leq R_2$.
    By positive homogeneity, $\alpha^\ast_i$ only takes values $0$ or $\infty$.
    It follows that $\alpha^\ast_i(Q)=0$ if and only if there exits $\lambda>0$ such that $dQ/dP\in \lambda C_i$ almost surely.
    Since $1$ has to be in $\lambda C_i$ for this to happen, we can constrain $1/\beta \leq \lambda \leq 1/\alpha$ in the case of $C_1$ and $1/(d\beta)\leq \lambda\leq 1/\alpha$ in the case of $C_2$.
    Thus
    \begin{equation*}
        \begin{split}
            R_1(X)&=\sup\left\{ E_Q\left[ X \right]\colon  \frac{dQ_k}{dP} \in \lambda C_1 \text{ for some } 1/\beta\leq \lambda\leq 1/\alpha \right\}\\
            R_2(X)&=\sup\left\{ E_Q\left[ X \right]\colon  \frac{dQ}{dP} \in \lambda C_2 \text{ for some } 1/(d\beta)\leq \lambda\leq 1/\alpha \right\}
        \end{split}
    \end{equation*}
\end{example}

\section{Risk Allocation}\label{sec02:riskallocation}

We have established in Theorem \ref{thm:rep} that the infimum over all allocations $m\in\mathbb{R}^d$ used for defining $R(X)$ is real valued and has the desired properties of a risk measure.
Beyond the question of the overall liquidity reserve, the allocation of this amount between the different risk components is key for systemic risk purposes.
We therefore address in this section the following questions:
\begin{itemize}
    \item The existence of a risk allocation;
    \item The uniqueness of a risk allocation;
    \item The impact of the interdependence structure,
\end{itemize}
The first question is important in some applications such as the default fund contribution of each member of a clearing house or the allocation of the capital among the different business lines of a bank.
As for the second question, non-uniqueness can become an issue when this allocation is a regulatory cost for the different members or desks.
If no additional clear rule is provided, the members would then face arbitrariness as for their contributions for the same overall risk.
As for the last question, systemic risk should reflect the level of dependence of the system.
For instance, highly correlated losses, while having the same marginal risk, should result into a higher systemic risk and different optimal allocations.
\begin{definition}
    A \emph{risk allocation} is an acceptable monetary allocation $m\in A(X)$ such that $R(X)=\sum m_k$.
    When a risk allocation is uniquely determined, we denote it by $RA(X)$.
\end{definition}
\begin{remark}
    By definition, if a risk allocation exists, then the full allocation property automatically holds; see also Section \ref{ss:sensi}.
\end{remark}
In contrast to the univariate case, where the unique risk allocation is given by $m=R(X)$, existence and uniqueness are no longer straightforward in the 
multivariate case.
The following example shows that existence may fail.
\begin{example}\label{ex:001}
    Consider the loss function $\ell(x,y)=x+y+(x+y)^+/(1-y)-1$ if $y<1$ and $\infty$ otherwise.
    It follows that $A(0)=\{ m \in \mathbb{R}^2\colon m_2>-1 \text{ and } 1\geq-m_1-m_2+(-m_1-m_2)^+/(1+m_2) \}$.
    Computations yield $R(0)=\inf_{m_2>-1}\{ m_2-(m_2^2+3m_2+1)/(m_2+2)\}=-1$.
    However, the infimum is not attained.
\end{example}
Our next result introduces conditions towards the existence and uniqueness of a risk allocation.
\begin{definition}
    We call a loss function $\ell$ \emph{permutation invariant} if, $\ell(x)=\ell(\pi(x))$ holds for every permutation $\pi$ of the components of the vector $x$.
\end{definition}
Note that the loss function used in Example \ref{ex:001} is not permutation invariant.
We denote by $Z=\{u\in \mathbb{R}^d\colon \sum u_k=0\}$ the set of \emph{zero-sum allocations}.
\begin{theorem}\label{thm:monalloc}
    If $\ell$ is a permutation invariant loss function, then, for every $X \in M^\theta  $, risk allocations $m^\ast$ exist.
    They are characterized by the first order conditions
    \begin{equation}\label{eq:lagrangecond}
        1 \in \lambda^\ast E\left[ \nabla \ell\left( X-m^\ast \right) \right]\quad \text{ and }\quad E\left[\ell\left(X-m^\ast \right)\right]=0,
    \end{equation}
    where $\lambda^\ast$ is a Lagrange multiplier.
    In particular, when $\ell$ has no zero-sum direction of recession\footnote{We refer the reader to Appendix \ref{appendix:01} regarding the notions and properties of recession cones and functions. In particular, if $\ell$ has no zero-sum direction of recession except $0$, then $\ell$ is an unbiased loss function.} except $0$, the set of the solutions $(m^\ast,\lambda^\ast)$ to the first order conditions \eqref{eq:lagrangecond} is bounded.
    
    If $\ell(x+\cdot)$ is strictly convex along zero-sums allocations for every $x$ with $\ell(x)\geq 0$, then the risk allocation is unique.
\end{theorem}
\begin{proof}
    Let $m$ in $A(X)$, according to Theorem \ref{thm:reccone}, it holds
    \begin{multline*}
        0^+A(X)=\left\{ u \in \mathbb{R}^d \colon E\left[ \ell\left( X-m-ru \right) \right]\leq 0, \text{ for all }r>0 \right\}\\
        = \left\{ u \in \mathbb{R}^d\colon \sup_{r>0}E\left[ \frac{\ell(X-m-ru)-\ell(0)}{r} \right]\leq 0\right\}\\
        = \left\{ u \in \mathbb{R}^d\colon E\left[ \sup_{r>0}\frac{\ell(X-m-ru)-\ell(0)}{r} \right]\leq 0\right\}=-0^+\ell 
    \end{multline*}
    Further, we define $f(m)=\sum m_k+\delta(m|A(X))$.
    It follows that $f$ is increasing, convex, lower semi-continuous, proper and such that $R(X)=\inf f$.
    Since $\ell(x)\geq \sum x_k-c$ and $R(X)>-\infty$, for $b \in A(X)$ it holds
    \begin{equation*}
        -\infty<R(X)\leq \sum b_k+r\sum u_k \leq \gamma <\infty\quad \text{ and }\quad b +r u\in A(X)
    \end{equation*}
    showing that $0^+f = Z\cap 0^+A(X)=-Z\cap 0^+\ell$.
    By \citep[Theorem 27.1 (b)]{rockafellar1970}, the existence of a risk allocation follows from $f$ being constant along its directions of recession $0^+f$, which according to Theorem \ref{thm:reccone}, is equivalent to $u \in 0^+f$ implies $(-u) \in 0^+f$.
    However, since $\ell$ is permutation invariant it follows that $0^+\ell=-0^+\ell$ and therefore $u\in 0^+f$ implies that $-u \in 0^+f$.
    Thus the existence of a risk allocation.\footnote{Note that this computation shows that the condition $Z\cap 0^+\ell=-Z\cap 0^+\ell$ is sufficient  to get the existence of a risk allocation.}
    In particular, if $0^+\ell=0$, then by \citep[Theorem 27.1, (d)]{rockafellar1970}, the set of risk allocations is non-empty and bounded.
    Furthermore, since $E[\ell(X-m)]<0$ for some $m$ large enough, the Slater condition for the convex optimization problem $R(X)=\inf_m f(m)$ is fulfilled.
    Hence, according to \citep[Theorems 28.1, 28.2 and 28.3]{rockafellar1970}, optimal solutions $m^\ast$ are characterized by \eqref{eq:lagrangecond}.
    
    Finally, let $m\neq n$ be two risk allocations.
    It follows that $\alpha m+(1-\alpha)n$ is a risk allocation as well for every $\alpha \in [0,1]$.
    Furthermore, $(m-n)$ is a zero sum allocation.
    By convexity, it follows that $0=E[\ell(X-\alpha m -(1-\alpha)n)]\leq \alpha E[\ell(X-m)]+(1-\alpha)E[\ell(X-n)]=0$ for every $0\leq \alpha\leq 1$, which shows that $\alpha \ell(X-m)+(1-\alpha)\ell(X-n)=\ell(X-\alpha m-(1-\alpha)n)$ $P$-almost surely for every $0\leq \alpha\leq 1$.
    Since $\ell(x+\cdot)$ is strictly convex on $Z$ for every $x$ such that $\ell(x)\geq 0$, it follows that $P[\ell(X-\alpha m-(1-\alpha)n)<0]=1$ for every $0\leq \alpha\leq 1$, showing in particular that $E[\ell(X-m)]<0$, a contradiction.
\end{proof}
\begin{corollary}\label{cor:transI}
    Let $\ell$ be a permutation invariant loss function, such that $\ell(x+\cdot)$ is strictly convex along zero-sum allocations for every $x$ with $\ell(x)\geq 0$.
    It holds
    \begin{equation*}
        RA(X+r)=RA(X)+r,\quad \text{for every }X \in M^\theta\text{ and }r \in \mathbb{R}^d.
    \end{equation*}
    If $\ell$ is additionally positive homogeneous, it holds
    \begin{equation*}
        RA\left( \lambda X \right)=\lambda RA(X), \quad \text{for every }X\in M^{\theta}\text{ and }\lambda>0
    \end{equation*}
\end{corollary}
\begin{proof}
    From Theorem \ref{thm:monalloc}, the assumptions on $\ell$ ensure the existence and uniqueness of a risk allocation uniquely characterized, together with the Lagrange multiplier, by the first order conditions.
    Let $m=RA(X+r)$, for which there exists a unique $\lambda$ such that $\lambda E\left[ \nabla\ell\left( X+r-m \right) \right]=1$ and $E[ \ell( X+r-m ) ]=c$.
    Hence, $n=m-r$ and $\lambda$ satisfy the first order conditions $\lambda E[ \nabla\ell( X-n ) ]=1$ and $E[ \ell( X-n )]=c$, which by uniqueness shows that $n=RA(X)=m-r=RA(X+r)-r$.
    As for the second assertion, it follows from $A(\lambda X)=\lambda A(X)$ for every $\lambda>0$ according to Proposition \ref{prop:01}.

\end{proof}
\begin{remark}
    In general, the positivity of the risk allocation is not required.
    However, if positivity or any other convex constraint is imposed, for instance by regulators, it can easily be embedded in our setup.
    In case of positivity, this would modify the definition of $R(X)$ into
    \begin{equation*}
        R(X)=\inf\left\{ \sum m_k\colon E\left[ \ell(X-m) \right] \leq 0 \text{ and }m_k\geq 0\text{ for every }k\right\},
    \end{equation*}
    with accordingly modified first order conditions.
\end{remark}
As already mentioned, the following example illustrates the importance of the uniqueness.
\begin{example}\label{exep:01} 
    Any loss function of class \ref{C1}, that is, $\ell(x)=h(\sum x_k)$, is permutation invariant.
    Thus, a risk allocation $m^\ast \in A(X)$ exists by means of Theorem \ref{thm:monalloc}.
    However, for any zero-sum allocation $u$, we have $R(X)=\sum m_k^\ast+u_k=\sum m_k^\ast$ and $E[h(\sum X_k-(m^\ast_k+u_k))]=E[h(\sum X_k-m^\ast_k)]\leq c$, so that $m^\ast+u$ is another risk allocation.

    In terms of regulatory costs, this is a problematic situation.
    Indeed, consider two banks and require from them $110$ M \euro\, and $500$ M \euro, respectively, as capital allocation.
    In such a case, one could equally well require $610$ M \euro\, from the first bank and nothing from the second.
    Such arbitrariness is unlikely to be accepted in that case.
\end{example}
Example \ref{exep:01} shows that loss functions of the class \ref{C1} lack the uniqueness of a risk allocation.
By contrast, for loss functions of class \ref{C2}, that is, $\ell(x)=\sum h(x_k)$, the following proposition shows that, while there exists a unique risk allocation under very mild conditions, the risk allocation only depends on the marginal distributions of the loss vector $X=(X_1,\ldots,X_d)$.
In other words, the risk measure and the risk allocation do not reflect the dependence structure of the system.
\begin{proposition}\label{prop:marginals}
    Let $\ell(x):=\sum h_k(x_k)$ for univariate loss functions $h_k:\mathbb{R}\to (-\infty,\infty]$ strictly convex on $\mathbb{R}_+$, $k=1,\ldots,d$.
    For every $X \in M^\theta$, there exists a unique optimal risk allocation $RA(X)$ and we have $RA(X)=RA(Y)$, for every $Y \in M^\theta$ such that $Y_k$ has the same distribution as $X_k ,$ $k=1,\ldots,d$.
\end{proposition}
\begin{proof}
    Let $x,y$ be such that $\alpha x +(1-\alpha)y \not \in \mathbb{R}^d_{-}$ for every $\alpha \in (0,1)$.
    It follows that $\ell(\alpha x+(1-\alpha)y) =\sum h_k(\alpha x_k+(1-\alpha)y_k)<\sum \alpha h_k(x_k)+(1-\alpha)h_k(y_k)=\alpha \ell(x)+(1-\alpha)\ell(y)$.
    The loss function $\ell$ is furthermore unbiased.
    Indeed, for every zero-sum allocation $u$, assuming without loss of generality $u_1>0$, it follows that
    \begin{align*}
        \ell0^+(u)&\geq \lim_{r\to \infty} h_1(ru_1)/r+\sum_{k\geq 2}h_k(ru_k)/r\geq \lim_{r\to \infty}h_1(ru_1)/r+\sum_{k\geq 2} u_k=\infty
    \end{align*}
    since $h_1$ is strictly convex and $h_1(t)\geq t$.
    Hence, $\ell$ has no zero-sum direction of recession other than $0$.
    The strict convexity of $h_k$ yields, according to Theorem 
\ref{thm:monalloc}, the existence of a unique risk allocation for every $X \in 
M^{\theta}$.
    The first order conditions \eqref{eq:lagrangecond} are written as
    \begin{equation*}
        1\in\lambda E\left[ \partial h_k(X_k-m_k) \right], \quad k=1,\ldots,d,\quad\text{and}\quad \sum E\left[ h_k\left( X_k-m_k \right) \right]=c ,
    \end{equation*}
    which only depend on the marginal distributions of $X$.
\end{proof}
Following \citet{rueschendorf2004} we can characterise in terms of supermodular, directionally convex and upper orthant stochastic ordering the risk of positive dependence in terms of $\ell$.
For a function $f:\mathbb{R}^d\to \mathbb{R}$ we define
\begin{equation*}
    \Delta_{k,y} f(x)=f(x_0,\ldots, x_k+y_k, \ldots, x_d)-f(x), \quad x,y \in \mathbb{R}^d, k\in \{1,\ldots,d\}
\end{equation*}
We say that a continuous function $f:\mathbb{R}^d\to \mathbb{R}$ is 
\begin{itemize}[fullwidth]
    \item super-modular, if $\Delta_{k,y}\Delta_{l,y}f(x)\geq 0$ for every $1\leq k< l\leq d$;
    \item directionally convex, if $\Delta_{k,y}\Delta_{l,y}f(x)\geq 0$ for every $1\leq k\leq l\leq d$;
    \item $\Delta$-monotone, if $\Delta_{i_1,y}\ldots \Delta_{i_n,y} f(x)\geq 0$ for every $\{i_1,\ldots,i_n\}\subseteq \{1,\ldots,d\}$;
\end{itemize}
for every $x$ and $y$ in $\mathbb{R}^d$ with $y\geqslant 0$.
We denote by $\succcurlyeq^{sm}$, $\succcurlyeq^{dc}$ and $\succcurlyeq^{uo}$ the integral orders given by the respective class of functions.
We refer to \citep{rueschendorf2004} for a discussion of these orders in terms of dependence risk.
Note that $X\geqslant^{uo} Y$ if and only if $P[X\geqslant x]\geq P[Y\geqslant x]$ for every $x \in \mathbb{R}^d$.
\begin{proposition}
    The shortfall risk measure $R$ is monotone with respect with $\succcurlyeq^{sm}$, $\succcurlyeq^{dc}$ or $\succcurlyeq^{uo}$ whenever $\ell$ is super-modular, directionally convex, or $\Delta$-monotone, respectively.
\end{proposition}
\begin{proof}
    The assertion follows immediately from the fact that if $\ell$ is one of super-modular, directionally convex, or $\Delta$-monotone, so is $\ell(\cdot -m)$ for every $m$.
    Therefore if $X\succcurlyeq^x Y$, it follows that $E[\ell(X-m)]\geq E[\ell(Y-m)]$ showing that $A(Y)\subseteq A(X)$.
\end{proof}
\begin{remark}
    Any loss function of the form \ref{C1}, \ref{C2} and \ref{C3} are directionally convex and therefore super-modular.
    They are $\Delta$-monotone if $d=2$.
    As for the specific loss functions used in this paper in several places for illustration
    \begin{align*}
        \sum \frac{(x_k^+)^2}{2}+\alpha \sum_{k<j} x_k^+x_j^+\\
        \sum x_k^++\alpha \sum_{k<j} (x_j+x_j)^+
    \end{align*}
    are both directionally convex and $\Delta$-monotone.
    However, if $\alpha=0$ they are degenerated in terms of these monotonicity since $\Delta_{k,y}\Delta_{j,y}\ell(x)=0$ for every $k\neq j$.
    As soon as $\alpha>0$, these loss functions are strictly monotone on $\mathbb{R}^d_+$.
%
\end{remark}
\begin{remark}
    A loss function can be chosen in view of an a-priori list of wished properties in terms of risk measurement and allocation as the Proposition above mentioned.
    However, loss function may also arise in systemic risk problems as an intrinsic property of the system as presented by \citet{eisenberg2001} or recently by \citet{weber2016}.
\end{remark}
\begin{example}\label{subsec:exep01}
The following simple example shows the impact of the dependence in a simple case for a loss function
\begin{equation}\label{eq:loss03}
    \ell(x_1,x_2) =\frac{1}{1+\alpha}\left[\frac{1}{2}e^{2x_1}+\frac{1}{2}e^{2x_2}+\alpha e^{x_1}e^{x_2}\right]-1.
\end{equation}
that is $\Delta$-monotone and bivariate normal vector $X=(X_1,X_2)\sim \mathcal{N}(0,\Sigma)$ with $\Sigma =
\begin{bmatrix}
    \sigma_1^2            & \rho \sigma_1\sigma_2  \\
    \rho \sigma_1\sigma_2 & \sigma^2_2             \\
\end{bmatrix}$. 
Solving the first order conditions yield
\begin{align*}
    RA_i(X) & =\sigma^2_i+\frac{1}{2}SRC(\rho,\sigma_1,\sigma_2,\alpha) &  R(X)    & =\sigma_1^2+\sigma_2^2+SRC(\rho,\sigma_1,\sigma_2,\alpha),
\end{align*}
showing that the risk allocations are disentangled into the respective individual contributions $\sigma_i^2$, $i=1,2$, and a \emph{systemic risk contribution}
\begin{equation}\label{eq:systemic_risk_contribution}
   SRC=\ln\left(1+ \alpha e^{\rho \sigma_1\sigma_2-\frac{1}{2}(\sigma_1^2+\sigma_2^2)} \right),
\end{equation}
which depends on the correlation parameter $\rho$ and on the systemic weight $\alpha$ of the loss function.
Figure \ref{fig:01} shows the value of this systemic risk contribution as a function of $\rho$ and $\sigma_1$.
\begin{figure}[ht]
  \centering
      \includegraphics[width=0.8\textwidth]{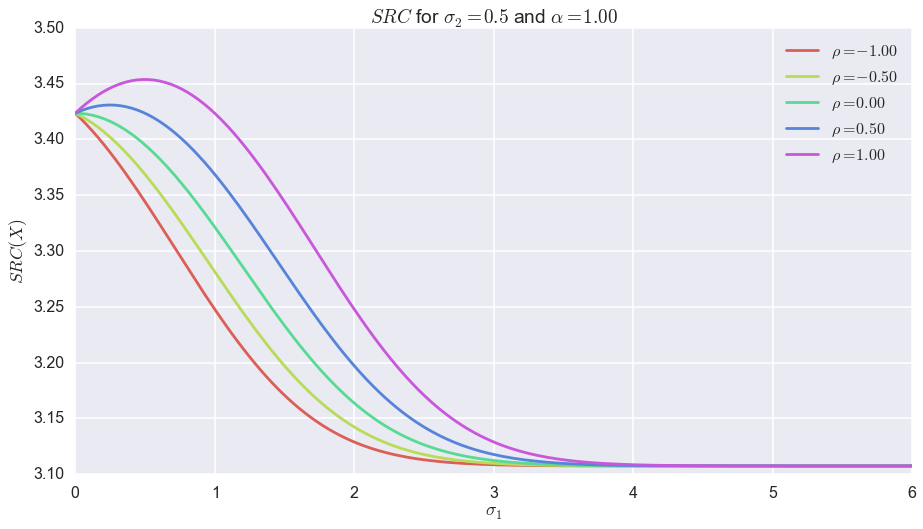}
      \caption{SRC \eqref{eq:systemic_risk_contribution} as a function of $\sigma_1$ for different values of the correlation $\rho$ in the case where $\alpha=1$.}
      \label{fig:01}
\end{figure}
Computing the partial derivatives with respect to $\sigma_i$ and $\rho$ yields
\begin{align*}
    \frac{\partial SRC}{\partial \sigma_1} & = \frac{\alpha\left( \rho\sigma_2-\sigma_1 \right)}{2}\frac{e^{\rho \sigma_1\sigma_2-\frac{1}{2}(\sigma_1^2+\sigma_2^2)}}{1+ \alpha e^{\rho \sigma_1\sigma_2-\frac{1}{2}(\sigma_1^2+\sigma_2^2)}}, &
    \frac{\partial SRC}{\partial \rho} & = \frac{\alpha \sigma_1\sigma_2}{2}\frac{ e^{\rho \sigma_1\sigma_2-\frac{1}{2}(\sigma_1^2+\sigma_2^2)}}{1+ \alpha e^{\rho \sigma_1\sigma_2-\frac{1}{2}(\sigma_1^2+\sigma_2^2)}}.
\end{align*}
showing that the systemic risk contribution is 
\begin{itemize}[fullwidth]
    \item increasing with respect to the correlation $\rho$;
    \item decreasing with respect to $\sigma_1$ if the correlation is negative;
    \item increasing up to $\rho\sigma_2$ and then decreasing with respect to $\sigma_1$ if the correlation is positive as the individual risk of $X_1$ dominates the risk of the system.
\end{itemize}

\end{example}


\section{Systemic Sensitivity of Shortfall Risk and its Allocation}\label{sec03:sensis}
The previous results emphasize the importance of using a loss function that adequately captures the systemic risk inherent to the system.
This motivates the study of the sensitivity of shortfall risk and its allocation so as to identify the systemic features of a loss function.
\begin{definition}
    The \emph{marginal risk contribution} of $Y \in M^{\theta}$ to $X \in M^{\theta}$ is defined as the sensitivity of the risk of $X$ with respect to the impact of $Y$, that is
    \begin{equation*}
        R(X;Y):=\limsup_{t \searrow 0}\frac{R(X+t Y)-R(X)}{t}.
    \end{equation*}
    In the case where $R(X+t Y)$ admits a unique risk allocation $RA(X+t Y)$ for every $t$, the \emph{risk allocation marginals} of the risk of $X$ with respect 
to the impact of $Y$ are given by
    \begin{equation*}
        RA_k(X;Y)=\limsup_{t \searrow 0}\frac{RA_k(X+t Y)-RA_k(X)}{t}, \quad k=1,\ldots , d.
    \end{equation*}
\end{definition}
Theorem \ref{thm:rep} and its proof show that the determination of the risk measure $R(X)$ reduces to the saddle point problem
\begin{equation*}
    R(X)=\min_{m}\max_{\lambda>0}L(m,\lambda,X)=\max_{\lambda>0}\min_m L(m,\lambda,X).
\end{equation*}
Using \citep{rockafellar1970}, the ``argminmax''  set of saddle points $(m^\ast,\lambda^\ast)$ is a product set that we denote by $B(X)\times C(X)$.

\begin{theorem}\label{thm:RC}
    Assuming that $\ell$ is permutation invariant, then 
    \begin{equation*}
        R(X;Y)=\min_{m\in B(X)}\max_{\lambda \in C(X)}\lambda E\left[ \nabla \ell \left( X-m \right)\cdot Y \right].
    \end{equation*}
    Supposing further that $\ell$ is twice differentiable and that $(m,\lambda) \in B(X)\times C(X)$ is such that 
    \begin{equation*}
        M=
        \begin{bmatrix}
            \lambda E\left[ \nabla^2 \ell (X-m) \right] & - 1/\lambda\\
            1 & 0
        \end{bmatrix}
    \end{equation*}
    is non-singular, then
    \begin{itemize}
        \item there exists $t_0>0$ such that $B(X+tY)\times C(X+tY)$ is a singleton, for every $0\leq t\leq t_0$;
        \item the corresponding unique saddle point $(m_t,\lambda_t)=(RA(X+tY),\lambda_t)$ is differentiable as a function of $t$ and we have
            \begin{equation*}
                \begin{bmatrix}
                    RA(X;Y)\\
                    \lambda(X;Y)
                \end{bmatrix}
                =M^{-1}V,
            \end{equation*}
            where $\lambda(X;Y) =\limsup_{t\searrow 0 }(\lambda_t-\lambda_0)/t$ and 
            \begin{equation*}
                V=
                \begin{bmatrix}
                    \lambda  E\left[ \nabla^2\ell(X-m)  Y \right]\\
                    R(X;Y)
                \end{bmatrix}.
            \end{equation*}
    \end{itemize}
\end{theorem}
\begin{proof}
    Let $L(m,\lambda,t)=\sum m_k + \lambda E[\ell(X+tY-m)]$.
    Theorem \ref{thm:rep} yields
    \begin{equation*}
        R(X+tY)=\min_m \max_{\lambda}L(m,\lambda,t)=\max_{\lambda} \min_m L(m,\lambda,t)=L(m_t,\lambda_t,t),
    \end{equation*}
    for every selection $(m_t,\lambda_t)\in B(X+tY)\times C(t+tY)$.
    Regarding the first assertion of the theorem, since $\ell$ has no zero-sum direction of recession other than $0$, it follows from Theorem \ref{thm:monalloc} that $B(X)\times C(X)$ is non empty and bounded.
    Hence, the assumptions of Golshtein's Theorem on the perturbation of saddle values, see \citet[Theorem 11.52]{rockafellar2009}, are satisfied and the first 
assertion follows.
    As for the second assertion, the assumptions of \citet[Theorem 6, pp. 34--45]{fiacco1990} are fulfilled.
    The Jacobian of the vector
    \begin{equation*}
        \begin{bmatrix}
            \nabla_m L(m,\lambda,0) \\
            \lambda E\left[ \ell\left( X-m \right)\right]
        \end{bmatrix}
    \end{equation*}
    that is used to specify the first order conditions is given by the matrix $M$.
    Hence, the second assertion follows from \citep[Theorem 6, pp. 34--35]{fiacco1990}.
\end{proof}
Theorem \ref{thm:RC} allows to explicitly derive the impact of an independent exogenous shock as stated in the following proposition.
\begin{proposition}\label{prop:causalresponsability}
    Under the assumptions of Theorem \ref{thm:RC} ensuring the uniqueness of a saddle point, suppose that $Y$ is independent of $X$.
    Then
    \begin{equation*}
        RC(X;Y)=\sum E\left[ Y_k \right] \quad \text{and}\quad RA(X;Y)=E[Y].
    \end{equation*}
\end{proposition}
\begin{proof}
    Since $Y$ is independent of $X$, denoting by $m=RA(X;Y)$, it follows from the first order conditions that
   \begin{equation*}
       RC(X;Y)=\lambda E\left[\nabla \ell(X-m)\cdot Y  \right]=\lambda E\left[ \nabla \ell(X-m) \right]\cdot E[Y]=1\cdot E[Y]=\sum E[Y_k]
   \end{equation*}
   Furthermore, we have
   \begin{equation*}
       M=
       \begin{bmatrix}
           \lambda A & -B\\
           C & 0
       \end{bmatrix}
       \quad \text{and}\quad
       V=
       \begin{bmatrix}
           \lambda E\left[ \nabla^2 \ell(X-m) Y\right] \\
           R(X;Y)
       \end{bmatrix}
       =
       \begin{bmatrix}
           \lambda A E[Y]\\
           C E[Y]
       \end{bmatrix}
   \end{equation*}
   where $A=E[\nabla^2\ell(X-m)]$ $B=\begin{bmatrix}
           1/\lambda & \cdots & 1/\lambda
      \end{bmatrix}^\intercal$, and $C=\begin{bmatrix}
          1 & \ldots & 1 
      \end{bmatrix}
   $.
   Using the classical formula of block matrix inversion, we obtain
   \begin{align*}
           RA(X;Y)&=
           \begin{bmatrix}
               \displaystyle  \frac{A^{-1}}{\lambda}-\frac{A^{-1}BCA^{-1}}{\lambda CA^{-1}B} &\displaystyle  \frac{A^{-1}B}{ CA^{-1}B}
           \end{bmatrix}
           \begin{bmatrix}
               \lambda A E[Y]\\
               C E[Y]
           \end{bmatrix}
           \\
           &= E\left[ Y \right]-\frac{A^{-1}BC E\left[ Y \right]}{CA^{-1}B}+\frac{A^{-1}BC E\left[ Y \right]}{ CA^{-1}B}=E\left[ Y \right].
   \end{align*}
\end{proof}
According to the discussion about causal responsibility in Section \ref{ss:sensi}, it follows that each member is marginally paying for the additional risk is takes provided this one is independent of the system.
In particular, if the risk factor $k$ is affected by a shock $Y_k$ independent of the system, it follows that $R(X;Y)=E[Y_k]=RA_k(X;Y)$, showing that the member $k$ pays for the full risks it takes.

\subsection{Impact of an Exogenous Shock}
The following Section illustrates the case when the exogenous shock may depend on $X$.
We consider a bivariate situation where $X=(X_1,X_3)$, and exogenous factor $Y=(Y_1,0)$ impacting only the first component.
We consider the loss function
\begin{equation*}
    \ell(x_1,x_2) =\frac{(x_1^+)^2+(x_2^+)^2}{2}+\alpha x_1^+x_2^+-1, \quad 0\leq \alpha \leq 1,
\end{equation*}
which gives rise to a unique risk allocation by virtue of Theorem \ref{thm:monalloc}.
Note that $\ell$ is $\Delta$-monotone, and strictly on $\mathbb{R}^2_+$ if $\alpha>0$.
For ease of notations, we assume that $X_1\sim X_2$, which, since $\ell$ is permutation invariant, implies that $m=RA_1(X)=RA_2(X)$.
Let $p=:P[X_1\geq m]=P[X_2\geq m]$ and $r=P[X_1\geq m;X_2\geq m]$.
According to Theorem \ref{thm:RC}, and the first order condition \eqref{eq:lagrangecond}, we have
\begin{equation*}
    R(X;Y)= \frac{E\left[ Y_1 (X_1-m_1)^+ \right]+\alpha pE\left[Y_1(X_2-m_2)^+|X_1\geq m_1\right]}{ E\left[ (X_1-m_1)^+ \right]+\alpha p E\left[(X_2-m_2)^+|X_2\geq m_2\right]}
\end{equation*}
As for the allocation of this marginal risk contribution, in the notation of Theorem \ref{thm:RC}, we have:
\begin{equation*}
    M=
    \begin{bmatrix}
        \lambda p          & \lambda \alpha r     & -1/\lambda\\
        \lambda \alpha r     & \lambda p          & -1/\lambda\\
        1                    & 1                    & 0
    \end{bmatrix}
    \quad \text{and}\quad
    V=
    \begin{bmatrix}
        \lambda pE\left[ Y_1 | X_1\geq m_1 \right]\\
        \lambda \alpha rE\left[ Y_1| X_1 \geq m_1; X_2 \geq m_2\right]\\
        R(X;Y)
    \end{bmatrix},
\end{equation*}
which by inverting $M$ yields
\begin{align*}
    RA_1(X;Y) & =\frac{R(X;Y)}{2} +\frac{1}{2}\frac{E\left[ Y_1 1_{\{X_1\geq m_1\}} \right]-\alpha E\left[ Y_11_{\{X_1\geq m_1; X_2\geq m_2\}} \right]}{p-\alpha r}\\
    RA_2(X;Y) & =\frac{R(X;Y)}{2} -\frac{1}{2}\frac{E\left[ Y_1 1_{\{X_1\geq m_1\}} \right]-\alpha E\left[ Y_11_{\{X_1\geq m_1; X_2\geq m_2\}} \right]}{p-\alpha r}
\end{align*}
Beyond the fact that according to Proposition \ref{prop:causalresponsability}, if $Y$ is independent of $X$ then $R(X;Y)=RA_1(X;Y)$ and $RA_2(X;Y)=0$, observe in general that:
\begin{itemize}[fullwidth]
    \item The two risk components marginally share first equally the additional cost of the exogenous impact in terms of $R(X;Y)/2$ each.
    \item The asymmetry of the shock that concerns only $X_1$ is reflected in the correction with respect to the second term which is added to the first one and subtracted to the second.
        Furthermore, $1_{\{X_1\geq m_1\}}\geq \alpha 1_{\{X_1\geq m_1;X_2\geq m_2\}}$ for every $0\leq \alpha\leq 1$.
        It implies that the additional risk taken by the first risk factor is always positively proportional to $Y_1$ while the second one is negatively proportional to $Y_1$.
    \item If $\alpha=0$, then the marginal change impact the risk factors according to $\pm (E[Y_1]-E[Y_1|X_1\geq m])/2$. 
    \item If $\alpha=1$ and $X_1$ and $X_2$ are strongly anti-correlated, then $1_{\{X_1\geq m;X_2\geq m\}}$ is likely very small and therefore the effect is similar to the case where $\alpha=0$.
        On the other hand, if $X_1$ and $X_2$ are strongly correlated, then $1_{\{X_1\geq m\}}\approx 1_{\{X_1\geq m;X_2\geq m\}}$ and in that case $RA_1(X;Y)\approx RA_2(X;Y)\approx R(X;Y)/2$ showing that the full dependence with $\alpha =1$ yields an equal share of the marginal risk changes.
\end{itemize}

\subsection{Sensitivity to Dependence}
Following the previous section where the loss function depends on $\alpha$ that impacts the risk allocation with respect to the degree of dependence between risk factors, we apply the techniques of Theorem \ref{thm:RC} to study the sensitivity with respect to $\alpha$.
To this end we consider a loss function of the following form
\begin{equation*}
    \ell(x)=\sum g(x_k)+\alpha h(x),
\end{equation*}
where $g$ is a one dimensional loss function and $h$ a multidimensional function such that $\ell$ is a loss function for all $\alpha\geq 0$ close enough to $0$.
For instance a loss function of the class \ref{C3}.
We also suppose that $g$ is twice differentiable.
Using the same strategy as in the proof of the Theorem \ref{thm:RC}, we can provide the marginal risk contribution and allocation as a function of $\alpha$ around $0$, stressing the dependence part of the loss function.
Computations yield
\begin{equation*}
    \partial_{\alpha} R(X)=\lambda E\left[ h(X-m) \right] \quad \text{and}\quad \partial_{\alpha}
    \begin{bmatrix}
        R(X)\\
        \lambda^{\prime}
    \end{bmatrix}
    =M^{-1}
    \begin{bmatrix}
        \lambda E\left[ \nabla h(X-m) \right] \\
        \partial_{\alpha}R(X)
    \end{bmatrix}
\end{equation*}
where $M$ is given by $M=
    \begin{bmatrix}
        \lambda A & -B \\
        C & 0
    \end{bmatrix}$ and $A=\text{diag}(g^{\prime\prime}(X_k-m_k))$ and $B$ and $C$ as in the proof of Proposition \ref{prop:causalresponsability}.
In the case where
\begin{equation*}
\ell(x)=\frac{1}{2}\sum_{k=1}^3 (x^+_k)^2+\alpha \sum_{1\leq k<j\leq 3}x_k^+x_j^+-1
\end{equation*}
and $X=(X_1,X_2,X_3)$ with $X_1\sim X_2\sim X_3$, $(X_1,X_2)\sim (X_2,X_1)$ and $X_3$ independent of $(X_1,X_2)$, it follows that $m=RA_k(X)$ for every $k=1,2,3$.
Defining $Z=(X_1-m)^+\sim (X_2-m)^+\sim (X_3-m)^+$, computations yields
\begin{equation*}
    \partial_{\alpha}R(X)= E[ Z ]\left(2+\frac{E[ (X_1-m)^+(X_2-m)^+ ]}{E[Z]^2}\right). 
\end{equation*}
Hence, with increasing correlation between $X_1$ and $X_2$ the marginal risk increases.
As for the impact on the risk allocation, since $E[(X_1-m)^+|X_2\geq m]=E[(X_2-m)^+|X_1\geq m]$ it simplifies to
\begin{equation*}
    \begin{split}
        \partial_{\alpha}RA_{1 \text{ or }2}(X) & = \frac{E[Z]}{3} \left(1 +\frac{E[ ( X_2-m)^+|X_1\geq m ]}{E[Z]}+\frac{E[ (X_1-m)^+(X_2-m)^+ ]}{E[Z]^2}\right)\\
        \partial_{\alpha}RA_3(X) & = \frac{E[Z]}{3} \left(4 -2\frac{E[ ( X_2-m)^+|X_1\geq m]}{E[Z]}+\frac{E[ (X_1-m)^+(X_2-m)^+ ]}{E[Z]^2}\right)
    \end{split}
\end{equation*}
Due to the asymmetric dependence of the system:
\begin{itemize}[fullwidth]
    \item One the one hand, if $X_1$ and $X_2$ are highly anti-correlated, then
        \begin{equation*}
            \partial_{\alpha}RA_{1\text{ or }2}(X)\approx \frac{E[Z]}{3}\quad \text{and}\quad \partial_{\alpha}RA_3(X)\approx 4\frac{E[Z]}{3}
        \end{equation*}
        The systemic risk factor is advantaging those who are anti-correlated, with respect to the others.
    \item On the other hand, if $X_1$ and $X_2$ are highly correlated, then for $p=P[X_1\geq m]$, 
        \begin{equation*}
            \partial_{\alpha}RA_{1\text{ or }2}(X)\approx \frac{E[Z]}{3}\left( \frac{p+1}{p}+\frac{E[Z^2]}{E[Z]^2}\right)\quad \text{while}\quad \partial_{\alpha}RA_3(X)\approx \frac{E[Z]}{3}\left(2\frac{p-1}{p}+\frac{E[Z^2]}{E[Z]^2}\right)
        \end{equation*}
        Since $p\leq 1$, the systemic risk factor penalizes those who are highly correlated and reduces the costs for the one who is independent with respect to the previous case.
\end{itemize}
Figure \ref{fig:sensialpha} illustrate this fact for different correlation values in the case of a 3-variate normal distribution
\begin{equation*}
    X\sim \mathcal{N}\left( 0, 
        \begin{bmatrix}
            1 & \rho & 0\\
            \rho & 1&0\\
            0 &0&1
        \end{bmatrix}
    \right)
\end{equation*}
\begin{figure}[ht]
  \centering
  \includegraphics[width=0.8\textwidth]{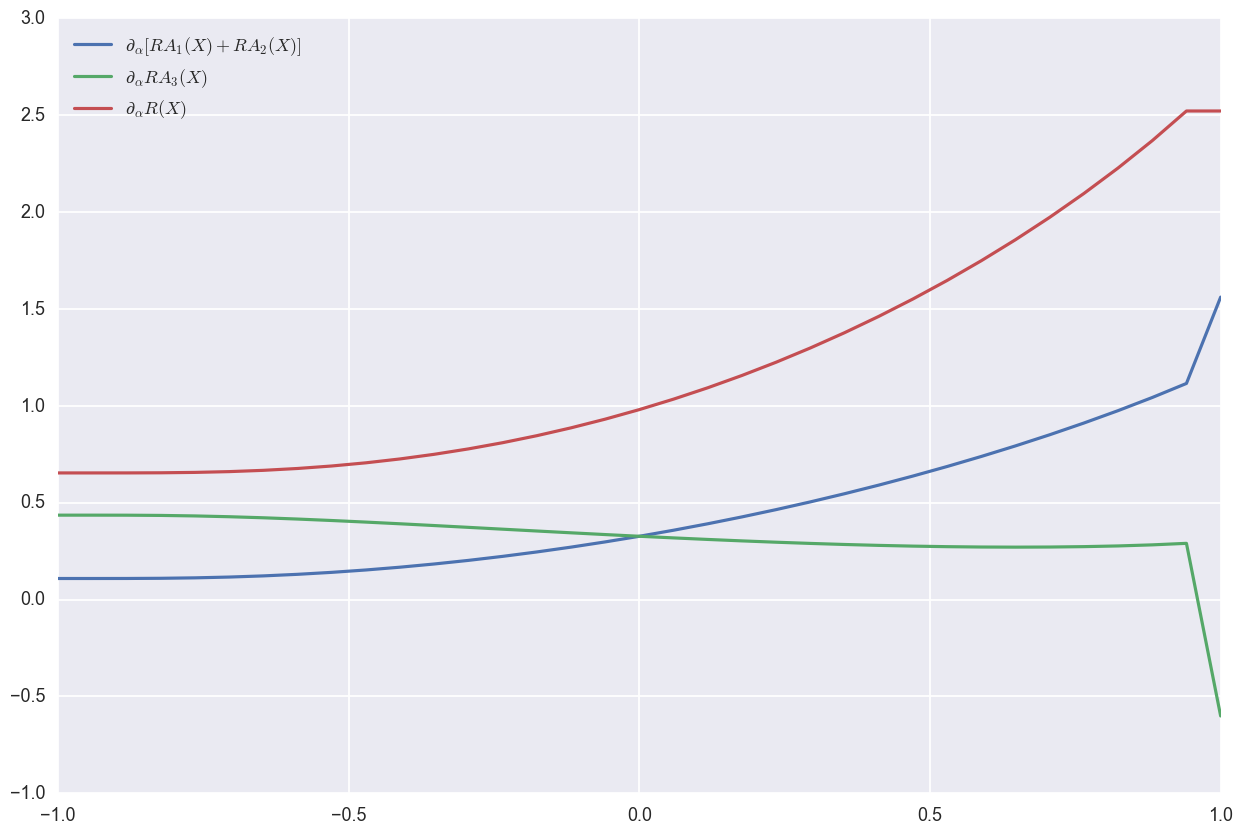}
  \caption{Systemic factor marginal change of the risk allocation and total risk for different correlations $\rho$.}
      \label{fig:sensialpha}
\end{figure}

\subsection{Riskless Allocation, Causal Responsibility and Additivity}\label{ss:sensi}

We conclude this section regarding risk allocation and its sensitivity by a discussion of their properties in light of the following economic 
features of risk allocations introduced in \citep{cheridito2014}.
\begin{enumerate}[label=\textbf{(FA)}]
    \item\label{fa} \textbf{Full Allocation:} $\sum RA_k(X)=R(X)$;
\end{enumerate}
\begin{enumerate}[label=\textbf{(RA)}]
    \item\label{ra} \textbf{Riskless Allocation:} $RA_k(X)=X_k$ if $X_k$ is deterministic;
\end{enumerate}
\begin{enumerate}[label=\textbf{(CR)}]
    \item\label{cr} \textbf{Causal Responsibility:}  $R(X+\Delta X_k)-R(X)=RA_k(X+\Delta X_k)-RA_k(X)$, where $\Delta X_k$ is a loss increment of the $k$-th risk component;
\end{enumerate}

As mentioned before, per design, shortfall risk allocations always satisfy the full allocation property \ref{fa}.
As visible from the above case studies, riskless allocation \ref{ra} and causal responsibility \ref{cr} are not satisfied in general.
In fact, from a systemic risk point of view, we think that \ref{ra} and \ref{cr} are not desirable properties.
Indeed, both imply that risk taking, or non-taking, should only impact the concerned risk component.
However, the risk components are interdependent and any move in one of them bears consequences to the rest of the system.
The search for an optimal allocation is a non-cooperative game between the different system components, each of them respectively looking for its own minimal risk allocation while impacting the others by doing so.
In other words, everyone is responsible for its own risk but also for its relative exposure with respect to the others.
The sensitivity analysis of this section however shows that external shocks are primarily born by the risk component that is hit at least in a first order.
In the case where this shock is independent of the system, by Proposition \ref{prop:causalresponsability} it is then a full causal responsibility.
Otherwise, a correction appears and a fraction of the shock is offloaded to the other risk components according to their relative exposure to the concerned component and dependence with the shock.

\section{Computational Aspects of Risk Allocation}\label{sec04:computation}

In this section we present computational results based on the loss function 
of Example \ref{ex:loss-functions}, that is, 
\begin{equation} \label{e:loss-function}
    \ell(x) = \sum_{k=1}^d x_k  + \frac{1}{2}\sum_{k=1}^d (x_k^+)^2 + \alpha \sum_{1 \le j<k \le d}   x_j^+ x_k^+-1,
\end{equation}
for $\alpha = 0$ or $1$.
In that case, the constrained problem \eqref{eq:def_shortfall_risk} becomes:
\begin{align} 
\begin{split}
    R(X) := \inf \Bigg\{ \sum m_k \colon
    & \sum_{k=1}^d E \left[ X_k - m_k \right] 
    + \frac{1}{2} \sum_{k=1}^d E \left[ \left( X_k - m_k \right)^+ \right]^2 \\
    & + \alpha \sum_{1 \le j<k \le d} E \left[ \left( X_j - m_j \right)^+ \left( X_k - m_k \right)^+ \right]
    \leq 1 \Bigg\}
\end{split}
\end{align}
According to Theorem \ref{thm:monalloc}, the risk allocation is determined by the first order conditions \eqref{eq:lagrangecond}, which read in this case:
\begin{equation}\label{e:loss-function-num-expect}
    \begin{cases}
        \displaystyle \lambda  E[(X_k-m_k)^+] + \alpha \lambda \sum_{j=1,j \neq k}^d E[(X_j-m_j)^+1_{\{X_k \geq m_k\}}]= 1 - \lambda, \quad \text{ for }\, k=1,\dots,d;\\
        \displaystyle \sum_{k=1}^d \left\{ E[(X_k-m_k)] + \frac{1}{2} E[((X_k-m_k)^+)^2] \right\} +\alpha\sum_{1\le j< k\leq d} E[(X_k-m_k)^+(X_j-m_j)^+] = 1.
    \end{cases}
\end{equation}
We use Gaussian distributions with mean vector $\mu$ and variance-covariance matrix $\Sigma$ for  the loss vector $X$.
In the bi- and tri-variate cases the variance-covariance matrix is parameterized by a single correlation factor $\rho$ and the variances $\sigma_k^2$ of $X_k$ for all $k$.
In other words, $\Sigma_{ij} = \rho \sigma_i \sigma_j$ for $i \neq j$.
We write CT for computational time.
The implementation was done on standard desktop computers in the Python programming language. To solve the constrained problem \eqref{eq:def_shortfall_risk}, we use the Sequential Least SQuares Programming (SLSQP) algorithm, in combination with Monte Carlo, Fourier or Chebychev interpolation schemes, briefly described below, for the computation of the expectations in \eqref{e:loss-function-num-expect}.

\paragraph{Fourier methods}
Assuming that the moment generating functions of the considered distributions are available, Fourier methods allow us to compute the different expectations in \eqref{e:loss-function-num-expect}, based on methods presented among other in \citet{EberleinGlauPapapantoleon08} and \citet{antonis2012} for details.
The main advantage of this method is that it is theoretically possible to compute the value of the integrals at any level of precision, while the basic computational time is roughly doubled for every additional digit of accuracy.
However, as seen in the subsequent computations this method suffers from the large number of double integrals to be computed, for which the computational time can become prohibitively long.

\paragraph{Monte-Carlo Methods}
We can also use Monte Carlo simulations for the estimation of the many integrals in \eqref{e:loss-function-num-expect}.
An important observation here is that we can generate and store all realizations in advance, and then use them for the estimation of the functions for different $m$ in every step of the root-finding procedure.
The main advantage of Monte Carlo relative to Fourier methods is that a wider variety of models can be considered; think, for example, of models with copulas or of random variables with Pareto type distributions as considered in the empirical study in Section \ref{sec05:empiricalstudy}.
The main disadvantage is the slow statistical convergence of the scheme, in our context though is fast enough.
In addition, the time to generate, once and for all, the samples, as well as to compute the Monte Carlo averages, is very fast and independent of the value of $m$.

\paragraph{Chebychev interpolation}
A numerical scheme well-suited to approximate the large numbers of functions in the context of optimization routines is the Chebyshev interpolation method.
This method, recently applied to option pricing by \citet{GassGlauMahlstedtMair15}, can be summarized as follows:
Suppose you want to evaluate quickly a function $F(m)$, of one or several variables, for a large number of $m$'s.
The first step of the Chebyshev method is to evaluate the function $F(m)$ on a given set of nodes $m^i$, $1\le i\le N$.
These evaluations can be computed by Fourier or Monte Carlo schemes, are independent of each other and can thus be realized in parallel.
The next step, in order to compute $F(m)$ for an $m$ outside the nodes $m^i$, is to perform a polynomial interpolation of the $F(m^i)$'s using the Chebyshev coefficients.
In other words, the Chebyshev method provides a polynomial approximation $\hat F(m)$ of $F(m)$. 

\paragraph{Discussion:}
Whether it is advantageous to use the Chebyshev interpolation or not, is a matter of two competing factors that affect the computational time:
On the one hand, the number of iterations $I(d)$ needed to find the root of the system and, on the other hand, the size of the grid $N^2$ used in the Chebyshev interpolation.
Our findings reveals that the Monte Carlo schemes are better than the Fourier schemes in the range of our accuracy requirements, since they require the least amount of work during each step of the root-finding procedure or for the pre-processing computations in the Chebyshev method.
Only when the dimension is low, less than three or $\alpha=0$ can the Fourier methods be faster.
Next, the choice between Chebyshev or not is a matter of comparison between $I(d)$ and $N^2$.
In high dimensions, when $I(d)$ dominates $N^2$, with $I(d)$ being in principle of order $d$ and $N$ usually between $10$ and $20$, then the Chebyshev method is less costly.
Furthermore, the Chebyshev method can intensively benefit from parallel computing as the pre-processing step is not sequential.

\subsection{Bivariate case}
We suppose that $d=2$ and consider a bivariate Gaussian distribution with zero mean, $\sigma_1 = \sigma_2 = 1$ and correlation $\rho$ in $\left\{-0.9,  -0.5, -0.2, 0, 0.2, 0.5,  0.9\right\}$.
When setting $\alpha=0$, that is without systemic risk weight, the result $m^*$ does not depend on the correlation value. 
Since $\sigma_1 = \sigma_2 = 1$ the allocation is symmetric and we find $m_1^* \approx -0.173$.
Explicit formulas for the involved expectations are available in this case and this yields of course the fastest computation.
Fourier methods are quite fast (CT $\approx 3 \times$ explicit formula) as we only need to compute 1-dimensional integrals. 
In order to get a high approximation in the Chebychev approximation, one must use 20 nodes for each integral. 
Since the number of iterations in the optimizations is about, the Chebychev method coupled with Fourier transforms is slower than Fourier without it.
Finally, Monte-Carlo is $\approx 20 \mbox{ to } 40 \times$ slower than Fourier, becoming the slowest method in that case.
When setting $\alpha = 1$, the values of the risk allocation are increasing with respect to $\rho$, as expected, see Table \ref{Num:Biv_correlation}. 
The Monte-Carlo method becomes the fastest one.
Indeed, we now need to compute bi-variate integrals in \eqref{e:loss-function-num-expect}.
Even if Fourier methods are fast (from 30 seconds to almost 3 minutes), they are still $\approx 10 \mbox{ to } 50$ slower than Monte-Carlo. 
Moreover, using even as little as 10 nodes in the Chebychev interpolation, which is not very accurate, increases the total computational time because of the number of 2-dimensional integrals to compute in the preprocessing step.
%

\begin{table}[ht] 
    \begin{center}
        \begin{tabular}{@{}cccccccccc@{}}
            \toprule
            
            & & \multicolumn{2}{c}{Fourier} & & \multicolumn{2}{c}{Fourier + Chebychev 10 nodes} & & \multicolumn{2}{c}{Monte Carlo 2 Mio} \\

            \cmidrule{3-4} \cmidrule{6-7} \cmidrule{9-10}

            \multicolumn{1}{c}{$\rho$} 
            & & \multicolumn{1}{c}{$m_1^\ast$} & \multicolumn{1}{c}{CT} 
            & & \multicolumn{1}{c}{$m_1^\ast$} & \multicolumn{1}{c}{CT}
            & & \multicolumn{1}{c}{$m_1^\ast$} & \multicolumn{1}{c}{CT} \\
            
            \midrule

            $ -0.9$   &  & -0.167 &  61520 ms & & -0.150 & 45 m 18 s & & -0.167 & 3257 ms \\            
            $ -0.5$   &  & -0.143 & 37100 ms & & -0.132 & 30 m 27 s & & -0.143 & 3357 ms \\
            $ -0.2$   &  & -0.120 & 45200 ms & & -0.113 & 25 m 21 s & & -0.120 & 3414 ms \\
            $    0$   &  & -0.103 & 51800 ms & & -0.098 & 24 m 52 s & & -0.103 & 3302 ms \\
            $  0.2$   &  & -0.085 & 75700 ms & & -0.082 & 27 m 55 s & & -0.085 & 3417 ms \\
            $  0.5$   &  & -0.057 & 158000 ms & & -0.055 & 32 m 10 s & & -0.056 & 3250 ms \\
            $  0.9$   &  & -0.013 & 88900 ms  & & -0.012 & 55 m 04 s & & -0.012 & 3387 ms \\
            \bottomrule
        \end{tabular}
        \caption{Bivariate case with systemic weight, that is, for $\alpha=1$.}
        \label{Num:Biv_correlation}
    \end{center}
\end{table}

\subsection{Trivariate Case} \label{ss:trivariate}
In this section, we illustrate the systemic contribution of the loss function with three risk components and study the impact of the interdependence of two components with respect to the third one.
We start with a Gaussian vector with the variance-covariance matrix
\begin{equation*}
    \Sigma = 
    \begin{bmatrix}
        0.5 & 0.5\rho &0\\
        0.5\rho & 0.5 &0\\
        0 & 0 &0.6\\
    \end{bmatrix},
\end{equation*}
for different correlations $\rho \in \left\{-0.9, -0.5, -0.2, 0, 0.2, 0.5, 0.9 \right\}$.
Here the third risk component has a higher marginal risk than the first two so that, in the absence of systemic component, it should contribute most to the overall risk.
When $\alpha = 0$, this is indeed the case.
The result is independent of the correlation and is typically overall lower, charging the risk component with the highest variance more -- $m_3^\ast \approx -0.12$ -- than the other two -- $m_1^\ast =m_3^\ast\approx -0.166$.
However, with systemic risk weight, the contribution of the first two overcomes the third one for high correlation, as emphasised in red in Table \ref{Num:triv_correlation}.
These results illustrate that the systemic risk weights correct the risk allocation as the correlation between the first two risk components increases.
The Monte Carlo scheme in this trivariate case is radically faster than Fourier, (and Chebychev interpolation was not found useful in this case either), from 30 times up to 60 times more efficient.

%
%
%
%

\begin{table}[ht]
    \begin{center}
        \resizebox{\columnwidth}{!}{
            \begin{tabular}{@{}ccccccccccccc@{}}
                \toprule
                
                & & \multicolumn{5}{c}{Fourier Method} & & \multicolumn{5}{c}{Monte Carlo 2 Mio} \\

                \cmidrule{3-7} \cmidrule{9-13}

                \multicolumn{1}{c}{$\rho$} & & 
                \multicolumn{1}{c}{$m_1^\ast=m_2^\ast$} & & \multicolumn{1}{c}{$m_3^\ast$} & \multicolumn{1}{c}{$R(X)$} & \multicolumn{1}{c}{TCP} & &
                \multicolumn{1}{c}{$m_1^\ast=m_2^\ast$} & & \multicolumn{1}{c}{$m_3^\ast$} & \multicolumn{1}{c}{$R(X)$} & \multicolumn{1}{c}{TCP} \\

                \midrule
                
                $-0.9$ &    & -0.189 & $\leq$ &  0.096 & -0.258 & 2 m 55 s &  
                            & -0.190 & $\leq$ &  0.095 & -0.283 & 3159 ms \\
                
                $-0.5$ &    & -0.135 & $\leq$ &  0.016 & -0.253 & 1 m 39 s &  
                            & -0.134 & $\leq$ &  0.017 & -0.252 & 2799 ms \\

                $-0.2$ &    & -0.099 & $\leq$ & -0.030 & -0.229 & 1 m 32 s &  
                            & -0.098 & $\leq$ & -0.030 & -0.228 & 2760 ms \\

                $0$    &    & -0.076 & $\leq$ & -0.059 & -0.212 & 2 m 22 s &  
                            & -0.077 & $\leq$ & -0.058 & -0.212 & 3188 ms \\

                $ 0.2$ &    & -0.053 & $\leq$ & -0.086 & -0.194 & 1 m 37 s &  
                            & -0.055 & $\leq$ & -0.086 & -0.195 & 2741 ms \\
                
                $0.5$  &    & -0.020 & $\textcolor{red}{\geq}$ & -0.125 & -0.165 & 1 m 47 s &  
                       & -0.020 & $\textcolor{red}{\geq}$ & -0.124 & -0.164 & 3358 ms \\

                $0.9$  &    &  0.025 & $\textcolor{red}{\geq}$ & -0.173 & -0.121 & 2 m 07 s &
                       &  0.026 & $\textcolor{red}{\geq}$ & -0.171 & -0.119 & 2722 ms \\

                \bottomrule
            \end{tabular}
        }
        \caption{Trivariate case with systemic weight, that is $\alpha=1$. Computed by Fourier.}
        \label{Num:triv_correlation}
    \end{center}
\end{table}

\subsection{Higher Dimensions}
 
Figure \ref{fig:30_dim} show the variance-covariance matrix and the resulting risk allocation in a 30-variate case using Monte Carlo, coupled with 15 node Chebychev interpolation when $\alpha=1$.
Indeed, the dimension being large, the preprocessing time with Monte-Carlo to compute the Chebychev coefficients together with the computational time resulting from the root-finding with the polynom is lower than the raw Monte-Carlo root finding.
The plot shows that the risk allocation depends not only on the variance of the different risk components, but also, in the case where $\alpha=1$, on the corresponding dependence structure.
For instance, compare components 28 and 29 in the 10-variate case in Figure \ref{fig:30_dim}.
In the first case we observe that when $\alpha=0$, component 28 contributes more than 29, and conversely when $\alpha=1$.
The reason is that even if component 28 has a slightly higher variance, it is relatively less correlated than 29 to the components 2, 3, 6, 20 and 23 that have the highest variance, and thus are the most `dangerous' from the systemic point of view.
Hence, component 29 is more exposed than 28 in case of a systemic event.

%
\begin{figure}[ht]
    \centering
        \includegraphics[width=0.9\textwidth]{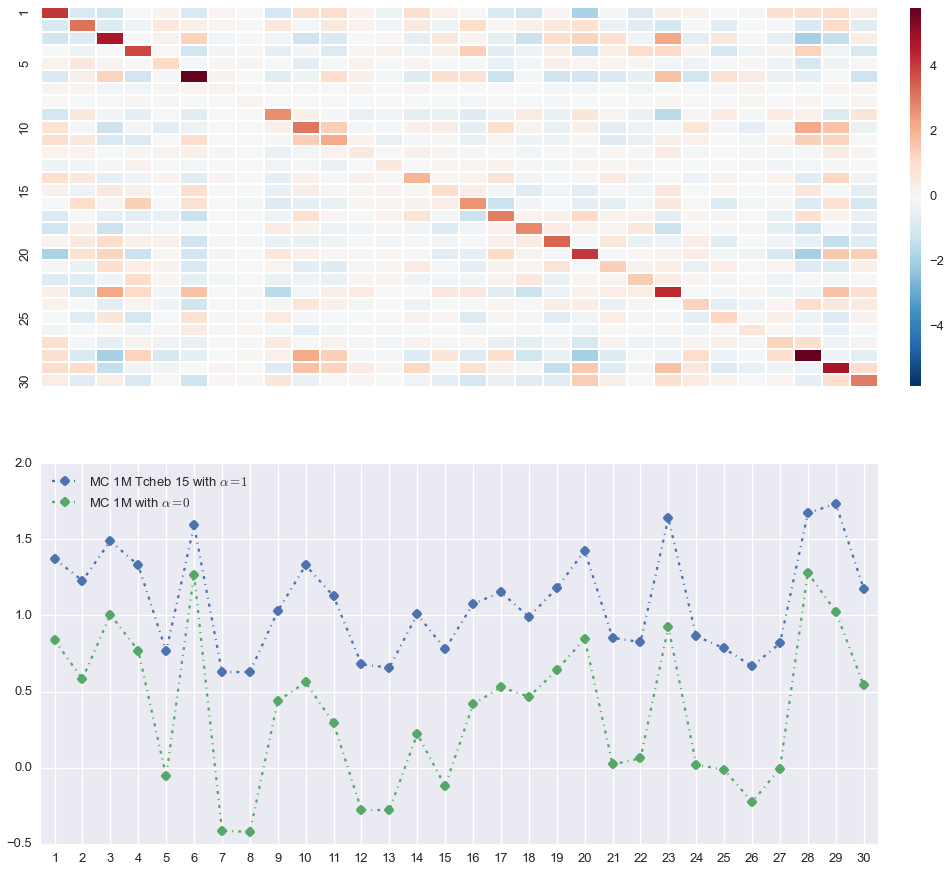}
        \caption{Plot showing the variance-covariance matrix together with the respective allocation in the 30-variate case for $\alpha=0,1$.}
        \label{fig:30_dim}
\end{figure}


\section{Empirical Study: Default Fund Allocation} \label{sec05:empiricalstudy}
{\def\X{X}
In the sequel we consider loss functions of the type
\begin{align}
    \ell_1(x) & = \sum x_k^+-\frac{1}{2}\sum x_k^- \label{eq:01}\\
    \ell_2(x) & = \sum x_k^+-\frac{1}{2}\sum x_k^-+\sum_{k\neq j} \left(x_k +x_j\right)^+-\frac{1}{2}\sum_{k\neq j} \left( x_k+x_j \right)^-\label{eq:03}
\end{align}
The first loss function means that a position is acceptable if on average, the losses are compensated by gains twice as large.\footnote{The coefficient $1/2$ is naturally subject to consensus and can be taken as any real number between $0$ and $1$.}
In this case, the risk assessment of the losses is marginal or component-wise.
The second one is similar, however, it also aggregates pairwise losses and gains among the different components.
Here the risk assessment considers additionally the pairwise dependence between the losses.
Note that each of these loss function is positive homogeneous (hence so is $R$) and permutation invariant.

The default fund of a CCP is a protection against extreme and systemic risk.
As of today, it is sized according to the Cover 2 rule, see  \citep[article 42, \S 3, p. 37]{europeanparlement2012}.
In a rough way, this corresponds to the maximal joint loss of two members over their posted collateral (initial margin) in a stressed situation over the last 60 days.
The relative contribution of each member to the default fund is proportional to their respective initial margin -- that is, the value at risk at a given level of confidence of their loss and  profit over a three-day time horizon.
Hence, denoting by $DF$ the total size of the default fund and by $IM_k(X_k)$ the initial margin of member $k$, the contribution of member $k$ is given by
\begin{equation} \label{eq:im_alloc}
    \frac{IM_k(X_k)}{\sum_j IM_j(X_j)} DF
\end{equation}
As an alternative, we propose to define the contribution of member $k$ to the default fund as follows.
According to theorem \ref{thm:monalloc} there exists a unique optimal capital allocation $RA(X)$ for a given loss vector $X$.
We define therefore the relative risk contribution of each financial component as
\begin{equation} \label{eq:risk_alloc}
    RC_k := RC_k(\X ) = \frac{RA_k(X)}{\sum RA_j(X)} = \frac{RA_k(X)}{R(X)}.
\end{equation}
The value at risk for the initial margins $IM_k$, the overall risk measure $R$ as well as the optimal capital allocation are all positive homogeneous.
It follows that $RC_k(\lambda X) = RC_k(X)$ for every $\lambda > 0$, that is, the relative risk contribution is scaling invariant as for instance the Sharpe ratio,
Minmax ratio or Gini ratio among others, see \citet{cheridito2013}.
The scaling invariance property allows one to consider the allocation independently of the total size of the default fund.
The contribution of member $k$ is then given as
\begin{equation} \label{eq:our_alloc_dollar}
    RC_k \times DF
\end{equation}
The current practice based on the ratio of initial margins \eqref{eq:im_alloc} provides an allocation that only depends on the marginal risk of each member profit and loss $X_k$, and does not take their joint dependence into account, that is, the systemic risk component.
By contrast, the approach \eqref{eq:our_alloc_dollar} allows one to take this systemic risk component into account in the allocation of the default fund in the sense of the following proposition already discussed in Section \ref{sec03:sensis}.

\subsection{Data}

In this section we compare a standard IM based allocation of the default fund of a CCP with the multivariate shortfall risk allocation resulting from the use of the loss functions $\ell_1$ and $\ell_2$. This empirical study is based on an LCH real dataset corresponding to the clearing of 74 portfolios of equity derivatives bearing on 90 underlyings. The clearing members have been anonymized and are referenced in the sequel by labels starting by PB plus number (e.g. PB7), whereas 
the underlying assets are identified by their real tickers, such as FCE for CAC40 index future and  AEX for Amsterdam exchange index, which can all be retrieved online.
The Jupyter notebook corresponding to this empirical study, including all the data and numerical codes, is publically available at https://github.com/yarmenti/MSRA. In order to avoid the repricing of the options, all the derivative positions have been linearized and reformulated in equivalent Delta positions in their underlyings. We denote by $P$ the $74 \times 90$ matrix of the positions of the 74 clearing members in the $90$ underlyings. As the CCP clears, each column of $P$ sums up to zero. The vector of the clearing member losses at a three day (3d) horizon is given by

\begin{equation} \label{e:\X }
	\X =-P\times (S_{3d}-S_0),
\end{equation}
where $S$ is the vector of the underlying price processes. The vector $S_{0}$ is observed and the vector $S_{3d}$ is simulated in a Student's t model estimated by maximum-likelihood on the underlying return time series, i.e.

\begin{equation} \label{e:Si}
	S_{3d}^i - S_{0}^i \sim \kappa_i \times T_i^{\nu_i} \times S_{0}^i,
\end{equation}
where $T_i^{\nu_i}$ is a Student's t random variable with $\nu_i$ degrees of freedom and where $\kappa_i$ a calibration fudge coefficient. 
The dependence between the underlyings  is modeled by a Student's t copula with correlation matrix $\rho$ and $\nu$ degrees of freedom, that is
\begin{equation*}
	C_{\rho, \nu}(u_1, \dots, u_n) = F_{\rho}^{\nu} \Big(F_{\nu}^{-1}(u_1), \dots, F_{\nu}^{-1}(u_n) \Big)
\end{equation*}
Here $F_{\rho}^{\nu}$ is the cumulative distribution function of the multivariate Student's t distribution with correlation matrix $\rho$ and $\nu$ degrees of freedom, and $F_{\nu}$ is the Student's t cdf with $\nu$ degrees of freedom.

\subsection{Simulations}

The correlation matrix $\rho$ is estimated empirically on the return time series and the dependence copula parameter is set to $\nu = 6$. Each of $m=10^5$ realizations of $S_{3d}$, hence of the loss vector $\X $, is simulated as follows:

\begin{enumerate}
	\item 	Simulate a Gaussian random vector $G$ of size 90 with zero-mean and correlation $\rho$
	\item 	Generate a $\chi_2$ random variable $\xi$ with parameter $\nu$
	\item 	Obtain the Student's t vector $R = \sqrt{\frac{\nu}{\xi}} G$
	\item 	Transform $R$ into uniform coordinates by $U_i = F_{\nu} \big(R_i \big)$ 
			and compute $T_i^{\nu_i}= F_{\nu_i}^{-1} \big( U_i \big)$
	\item 	Compute $S_{3d}$ by \eqref{e:Si}
and $\X $ by \eqref{e:\X }
\end{enumerate}
The resulting inputs to the allocation optimization problem are analysed in Appendix \ref{appendix:03}. 
Figure \ref{f:correl} shows the correlation matrices of the underlying assets and of the loss vector $\X $ of the clearing members, in a heatmap representation. In the left panel, which is directly estimated from the data, we see that the underlying assets are all positively correlated, as commonly found in the case of equity derivatives. However, due to positions in opposite directions taken by the clearing members, some of their losses exhibit significant negative correlations, as shown by the blue cells in the right panel.

\begin{figure}[htbp]
    \centering
	\includegraphics[width=0.45\textwidth]{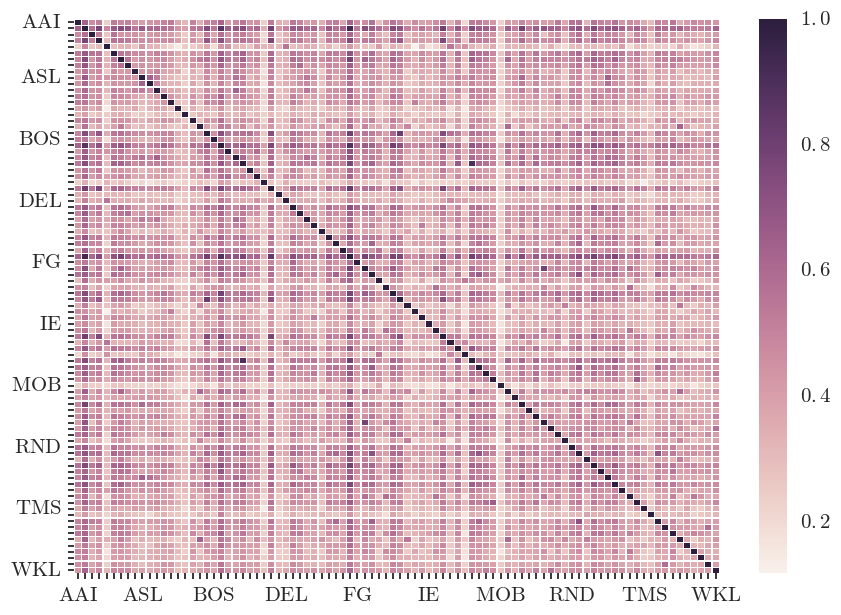} 
    \includegraphics[width=0.45\textwidth]{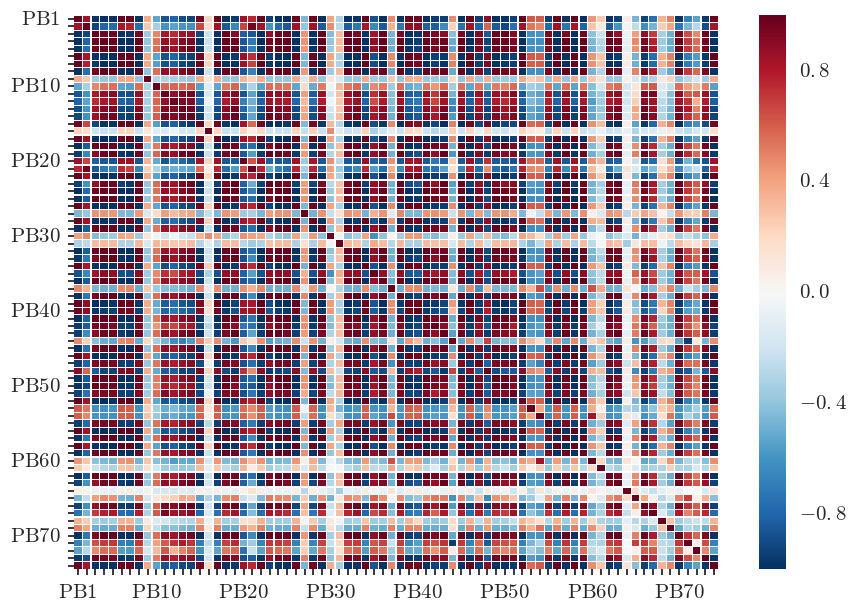} 
    \caption{ 	Left: Correlation matrix of the underlying assets (ranked by alphabetical order of asset ticker; one ticker out of ten is displayed along the coordinate axes). 
    			Right: Correlation matrix of the loss vector $\X $ of the clearing members (ranked by alphabetical order of member label; one label out of ten is displayed along the coordinate axes).
    		}
    \label{f:correl}
\end{figure}

\subsection{Allocation Results} \label{s:empi}

The total size of the default fund as of a standard Cover 2 methodology are shown in Table \ref{t:cover 2}, for three values of the dependence copula parameter $\nu$ and for 99\% vs. 99.7\% initial margins (IM).
Since a Cover 2 default fund is a cushion over IM, its size is directly responsive to the level of the quantile which is used for setting the IM (compare the two lines in Table \ref{t:cover 2}).  In relative terms the size of the default fund is quite stable with respect to $\nu$. However we emphasize that these are monetary amounts, so that the difference between for instance 6.16 $10^8$ and 6.72 $10^8$ corresponds to
0.56 $10^8$, i.e. more than half a billion of the corresponding currency. 
\begin{table}[ht]
	\begin{center}
		\begin{tabular}{@{}rcccc@{}}
			\toprule
			 			& & $\nu = 2$		& $\nu = 6$		& $\nu = 50$ \\
			\midrule
            99 \% IM	& & 6.16 $10^8$ 	& 6.72 $10^8$ 	& 6.27 $10^8$ \\
            99.7 \% IM	& & 4.96 $10^8$ 	& 5.48 $10^8$ 	& 5.00 $10^8$ \\
            \bottomrule
		\end{tabular}
		\caption{Size of a Cover 2 default fund for different levels of initial margins and different values of the dependence copula parameter $\nu$.}
		\label{t:cover 2}
	\end{center}
\end{table}
In the sequel we set $\nu=6,$ which corresponds to an intermediate level of tail dependence, and we use 99\% IM, which corresponds to the EMIR regulatory floor on initial margins.


Figure \ref{f:l1_im_alloc} compares the allocation weights implied by the loss function $\ell_1$ with the ones implied by 99\% IM. The allocations are very similar, as confirmed by the examination of the percentage relative differences displayed in the upper panels of Figure \ref{f:l1_im_alloc}.
By contrast, the lower panels of Figure \ref{f:diff_l1} show that the allocation weights implied by the loss function $\ell_1$ and the dependence sensitive loss function $\ell_2$ differ significantly in relative terms, including for the names with the greatest contributions to the default fund.
These results illustrate the impact of the use of a ``systemic'' loss function on the allocation of the default fund.

\begin{figure}[htbp]
	\centering
	\includegraphics[width=0.45\textwidth]{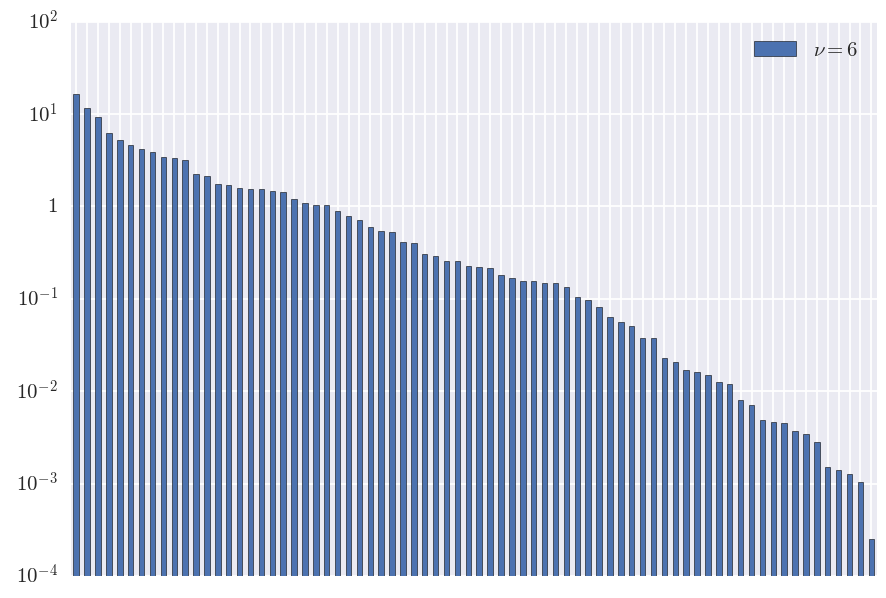} 
	\includegraphics[width=0.45\textwidth]{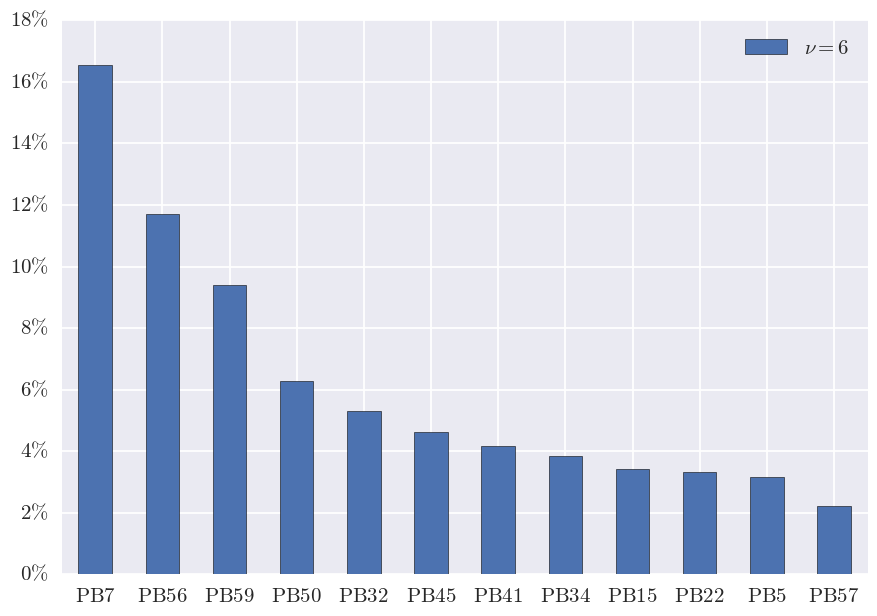} 

	\includegraphics[width=0.45\textwidth]{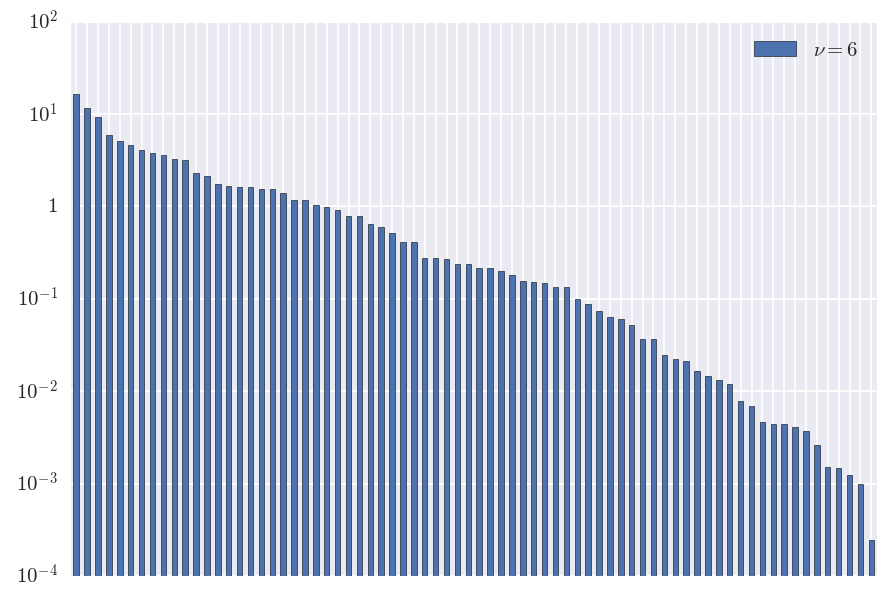}
	\includegraphics[width=0.45\textwidth]{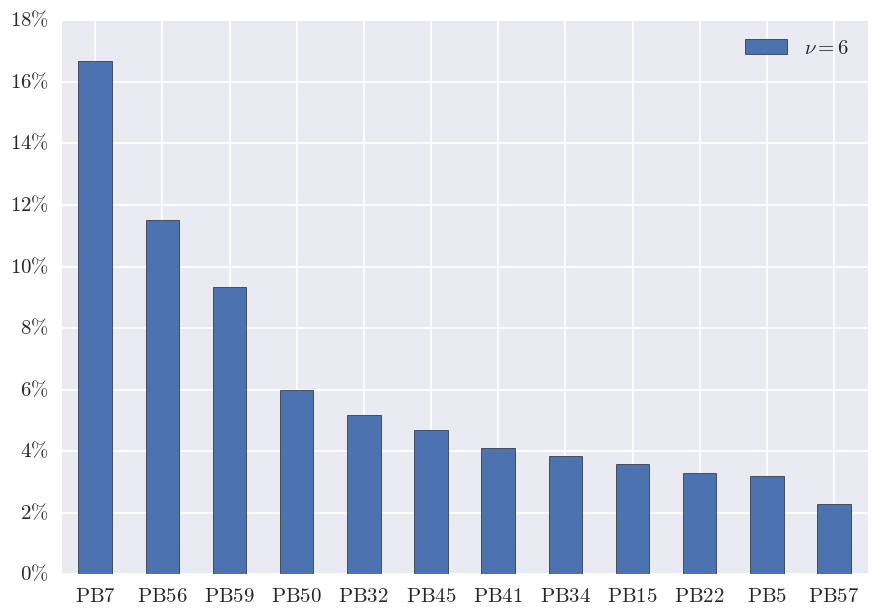}

	\caption{	Left: Decreasing log-allocation weights implied by the loss function $\ell_1$ (top) and 99\% IM (bottom).
				Right: Twelve highest allocation weights implied by the loss function $\ell_1$ (top) and by 99\% IM (bottom), 
				with the corresponding member labels.}

	\label{f:l1_im_alloc}
\end{figure}

\begin{figure}[htbp]
    \centering
    \includegraphics[width=0.45\textwidth]{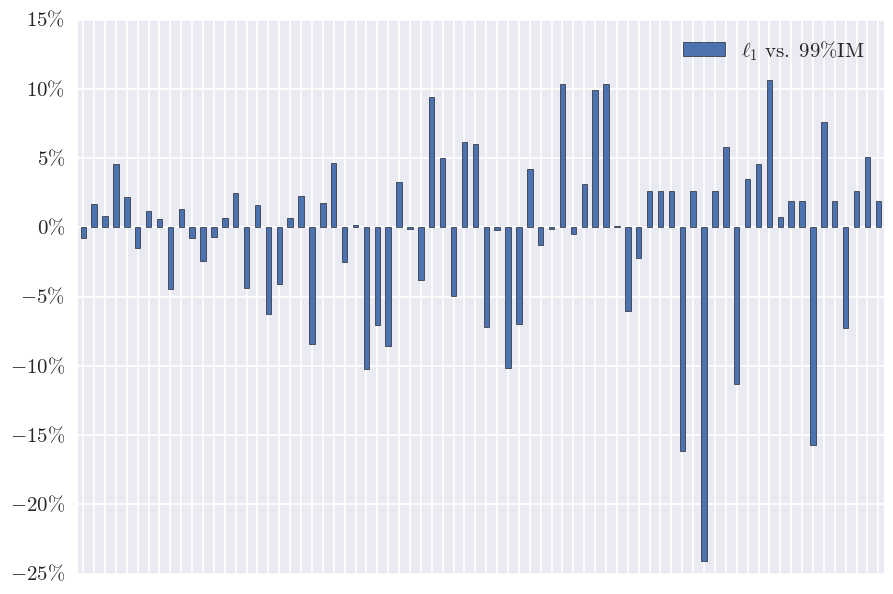}
    \includegraphics[width=0.45\textwidth]{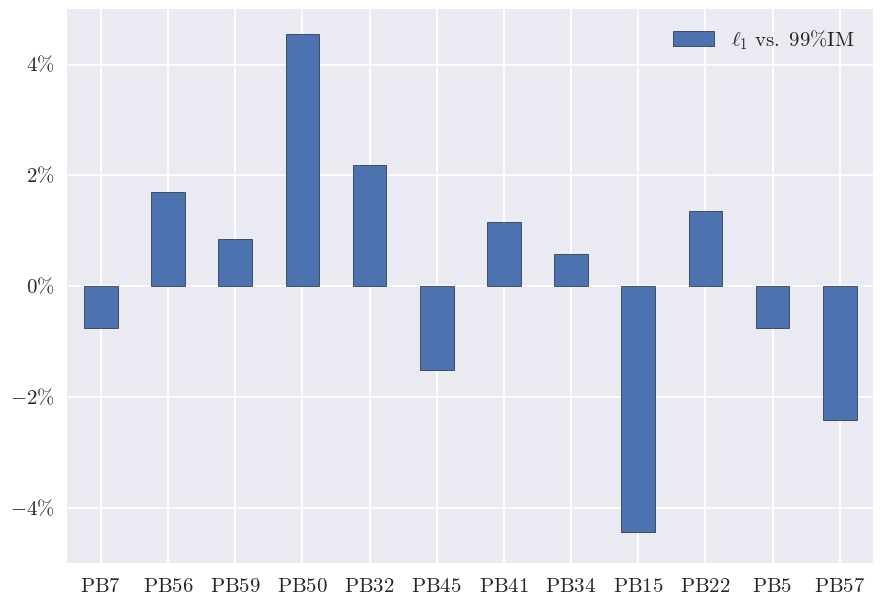}

    \includegraphics[width=0.45\textwidth]{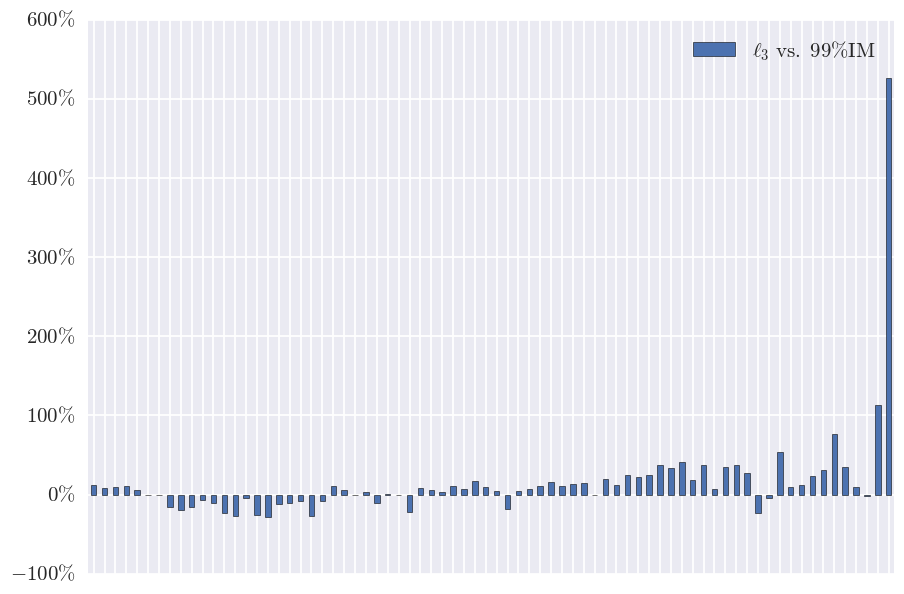}
    \includegraphics[width=0.45\textwidth]{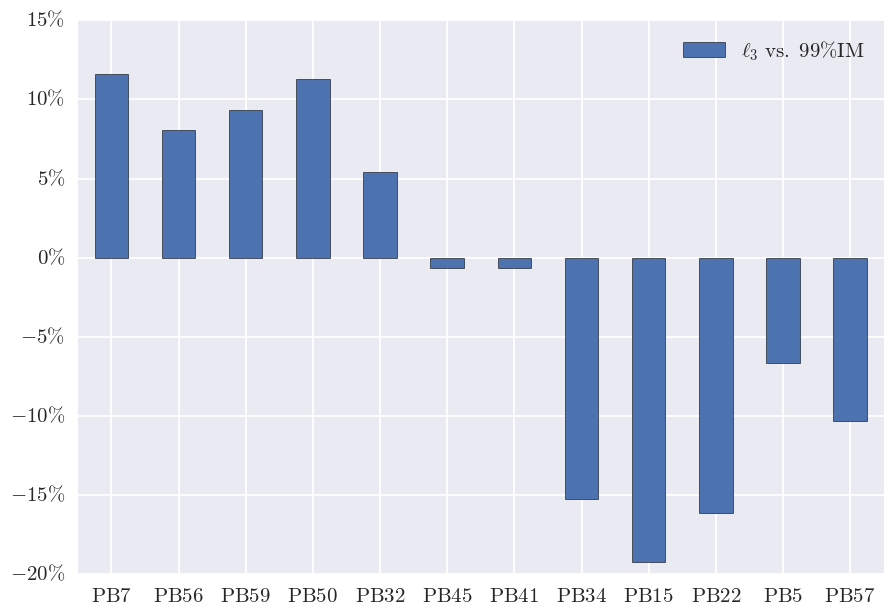}

    \includegraphics[width=0.45\textwidth]{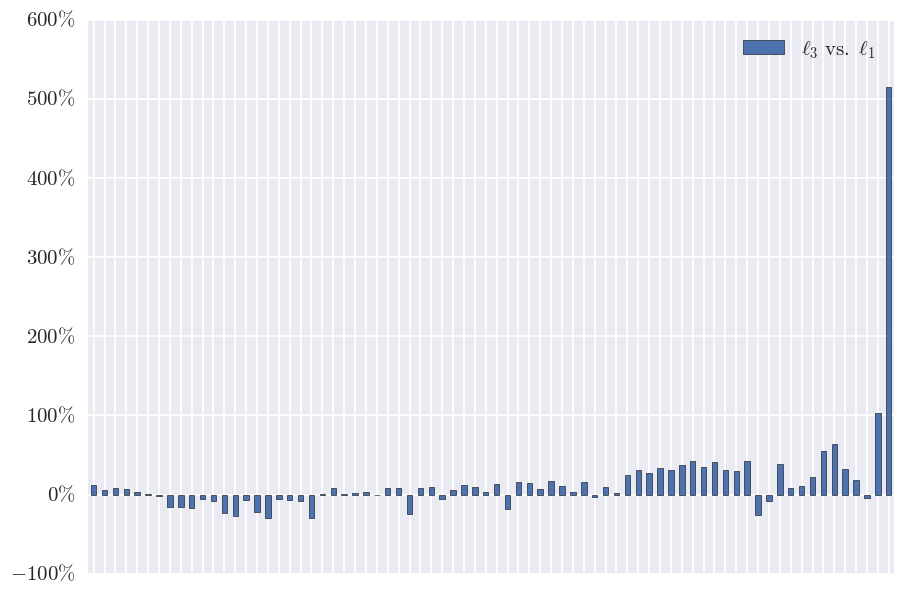}
    \includegraphics[width=0.45\textwidth]{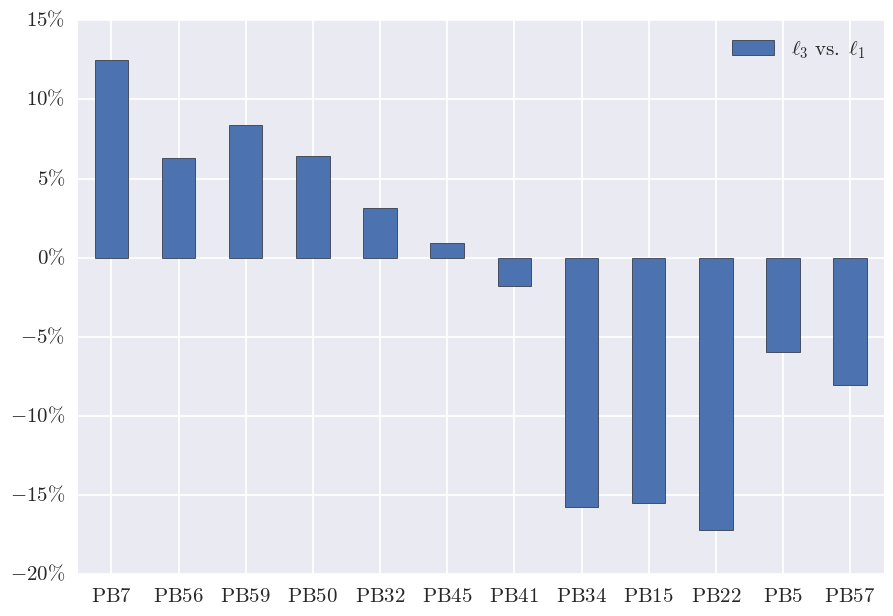} 
    
    \caption{	Left: Percentage relative differences between the allocation weights implied by 
                the loss function $\ell_1$ and 99\%IM (top), 
                the loss function $\ell_2$ and 99\% IM (middle),
                and the loss functions $\ell_1$ and $\ell_2$ (bottom), 
                ranked by decreasing values of the allocation weights implied by the loss function $\ell_1$.
    			Right: Zoom on the left parts of the graphs, with member labels.}

    \label{f:diff_l1}
\end{figure}}

\vskip0.5cm

\paragraph{Acknowledgements}
This paper greatly benefited from regular exchanges with the quantitative research team of LCH in Paris: Quentin Archer, Julien Dosseur, Pierre Mouy and Mohamed Selmi. In particular we are grateful to Pierre Mouy for the preparation of the real dataset used for the empirical study of Section \ref{sec05:empiricalstudy}.

\newpage
\begin{appendix}

\section{Some Classical Facts in Convex Optimization}\label{appendix:01}

For an extended real valued function $f$ on a locally convex topological vector space $X$, its convex conjugate is defined as
\begin{equation*}
    f^\ast(x^\ast)=\sup_{x \in X}\left\{ \langle x^\ast,x\rangle -f(x)\right\}, \quad x^\ast \in X^\ast,
\end{equation*}
where $X^\ast$ is the topological dual of $X$.
The Fenchel--Moreau theorem states that if $f$ is lower semi-continuous, convex and proper, then so is $f^\ast$, and it holds
\begin{equation*}
    f(x)=f^{\ast \ast}(x)=\sup_{x^\ast\in X^\ast}\left\{ \langle x^\ast, x\rangle -f^\ast(x^\ast)\right\}, \quad x \in X.
\end{equation*}
Following \citet{rockafellar1970}, for any non-empty set $C\subseteq \mathbb{R}^d,$ we define its \emph{recession cone}
\begin{equation*}
    0^+C:=\left\{ y \in \mathbb{R}^d\colon x+\lambda y\in C\text{ for every }x \in C \text{ and }\lambda\in \mathbb{R}_+ \right\}.
\end{equation*}
By \citep[Theorem 8.3]{rockafellar1970}, if $C$ is non-empty, closed and convex, then
\begin{equation}\label{eq:recesscone}
    0^+C=\left\{ y \in \mathbb{R}^d\colon \text{ there exists } x \in C\text{ such that }  x+\lambda y\in C\text{ for  every }\lambda\in \mathbb{R}_+\right\}.
\end{equation}
By \citep[Theorem 8.4]{rockafellar1970}, a non-empty, closed and convex set $C$ is compact if and only if $0^+C=\{0\}$.

Given a proper, convex  and lower semi-continuous function $f$ on $\mathbb{R}^d$, we call  $y\in\mathbb{R}^d$ a \emph{direction of recession} of $f$ if there exists $x \in \text{dom}(f)$ such that the map $\lambda \mapsto f(x+\lambda y)$ is decreasing on $\mathbb{R}_+$.
We denote by $f0^+$ the \emph{recession function} of $f$, that is, the function with epigraph given as the recession cone of the epigraph of $f$, and we call
\begin{equation*}
    0^+f:=\left\{y\in \mathbb{R}^d\colon (f0^+)(y)\leq 0\right\}
\end{equation*}
the \emph{recession cone of $f$}.
The following theorem gathers results from \citep[Theorems 8.5, 8.6, 8.7 and Corollaries pp. 66--70]{rockafellar1970}.
\begin{theorem}\label{thm:reccone}
    Let $f$ be a proper, closed and convex function on $\mathbb{R}^d$.
    \begin{enumerate}
        \item Given $x,y$ in $\mathbb{R}^d$, if $\liminf_{\lambda \to \infty}f(x+\lambda y)<\infty,$  then $\lambda \mapsto f(x+\lambda y)$ is decreasing.
        \item All the non-empty level sets $B:=\{x \in \mathbb{R}^d \colon f(x)\leq \gamma\}\neq \emptyset$ of $f$ have the same recession cone, namely the recession cone of $f$. That is:
            \begin{equation*}
                0^+f=0^+B,  \text{ for every }\gamma \in \mathbb{R}\text{ such that }B \neq \emptyset.
            \end{equation*}
        \item $f0^+$ is a positively homogeneous, proper, closed and convex function, such that
            \begin{equation*}
                (f0^+)(y)=\sup_{\lambda > 0}\frac{f(x+\lambda y)-f(x)}{\lambda}=\lim_{\lambda \to \infty}\frac{f(x+\lambda y)-f(x)}{\lambda}, \quad y \in \mathbb{R}^d,
            \end{equation*}
            for every $x \in \text{dom}(f)$.
        \item There exists $x \in \text{dom}(f)$ such that the map $\lambda \mapsto f(x+\lambda y)$ is decreasing on $\mathbb{R}_+$, that is, $y$ is a direction of recession of $f$, if and only if this map is decreasing for every $x\in \text{dom}(f),$ which in turn is equivalent to $(f0^+)(y)\leq 0$.
        \item The map $\lambda \mapsto f(x +\lambda y)$ is constant on $\mathbb{R}_+$ for every $x\in \text{dom}(f)$ if and only if $(f0^+)(y)\leq 0$ and $(f0^+)(-y)\leq 0$.
    \end{enumerate}
\end{theorem}

\section{Multivariate Orlicz Spaces}\label{appendix:02}

In this appendix we briefly sketch how the classical theory of univariate Orlicz spaces carries over to the $d$-variate case without any significant change.
We follow the lecture notes by \citet{leonard2007}, only providing the proofs that differ structurally from the univariate case.

A function $\theta :\mathbb{R}^d \to [0,\infty]$ is called a Young function if it is
\begin{itemize}
    \item convex and lower semi-continuous;
    \item such that $\theta(x)=\theta(\abs{x})$ and $\theta(0)=0$;
    \item non trivial, that is, $\text{dom}(\theta)$ contains a neighborhood of $0$ and $\theta(x)\geq a\norm{x}-b$ for some $a>0$.
\end{itemize}
In particular, $\theta$ achieves its minimum at $0$ and is increasing on $\mathbb{R}^d_+$.
It is said to be finite if $\text{dom}(\theta)=\mathbb{R}^d$ and strict if $\lim_{x\to \infty} \theta(x)/\norm{x}=\infty$.
\begin{lemma}
    The function $\theta$ is Young if and only if $\theta^\ast$ is Young.
    Furthermore, $\theta$ is strict if and only if $\theta^\ast$ is strict if and only if $\theta$ and $\theta^\ast$ are both finite.
\end{lemma}
\begin{proof}
    This follows by application of the Fenchel-Moreau theorem and from the relation $x \cdot y\leq \theta(x)+\theta^\ast(y)$.
\end{proof}
For $X \in L^0$, the Luxembourg norm of $X$ is given as
\begin{equation*}
    \norm{X}_{\theta}=\inf \left\{ \lambda\in \mathbb{R} \colon \lambda >0 \text{ and  }E\left[ \theta\left( X/\lambda \right) \right]\leq 1\right\},
\end{equation*}
where $\inf \emptyset =\infty$.
The Orlicz space and heart are respectively defined as
\begin{align*}
    L^\theta & :=\left\{ X \in L^0\colon \norm{X}_{\theta}<\infty \right\}=\left\{ X \in L^0\colon E\left[ \theta\left( X/\lambda \right) \right]<\infty \text{ for some  }\lambda \in \mathbb{R}, \lambda>0\right\}\\\nonumber
    M^{\theta} &:=\left\{ X \in L^0\colon E\left[ \theta\left( X/\lambda \right) \right]<\infty \text{ for all  }\lambda \in \mathbb{R}, \lambda>0\right\}.
\end{align*}
\begin{lemma}\begin{enumerate}
        \item We have $\norm{X}_{\theta}=0$ if and only if $X=0$.
        \item If $0< \norm{X}_{\theta}<\infty$, then $E[\theta(X/\norm{X}_\theta)]\leq 1$.
            In particular, $B:=\{X\colon \norm{X}_{\theta}\leq 1\}=\{X\colon E[\theta(X)]\leq 1\}$.
        \item The gauge $\norm{\cdot}_{\theta}$ is a norm both on the Orlicz space $L^\theta$ and on the Orlicz heart $M^\theta$.
        \item The following H\"older Inequality holds:
            \begin{equation*}
                E\left[ \abs{X \cdot Y} \right]\leq \norm{X}_{\theta}\norm{Y}_{\theta^\ast} .
            \end{equation*}
        \item $L^{\theta}$ is continuously embedded into $L^1$, the space of integrable random variables on $\Omega\times \{1,\ldots,d\}$ for the product measure $P\otimes \text{Unif}_{\{1,\ldots,d\}}.$
            \footnote{The case where $L^{\theta}=L^1$ corresponds to $\theta (x)=\sum |x_k|.$}
        \item The normed spaces $(L^\theta, \norm{\cdot}_{\theta})$ and $(M^\theta, \norm{\cdot}_{\theta})$ are Banach spaces.
    \end{enumerate}
\end{lemma}
\begin{proof}
    These results can be established along the same lines as in the univariate case \citep[See][Lemmas 1.8 and 1.10 and Propositions 1.11, 1.14, 1.15 and 1.18]{leonard2007}, using the Fenchel-Moreau Theorem in $\mathbb{R}^d_+$.
\end{proof}

\begin{theorem}\label{thm:duality}
    If $\theta$ is finite, then the topological dual of $M^{\theta}$ is $L^{\theta^\ast}$.
\end{theorem}
\begin{proof}
    Again, the proof follows the univariate case \citep[see][Proposition 1.20, Theorem 2.2 and Lemmas 2.4 and 2.5]{leonard2007}.
\end{proof}

\section{Data Analysis} \label{appendix:03}

Figure \ref{f:gross_pos_cm} shows the gross positions (sum of the absolute values of the positions in the underlying asset) per clearing member. Four members concentrate particularly high positions in the CCP.
Figure \ref{f:gross_pos_udl} shows the gross positions of the CCP per underlying asset (top) and the corresponding underlying asset values (bottom). The largest investment by far of the clearing members is in the asset with ticker FCE 
(CAC40 index future, with spot value 4463), by a factor about three to the second one AEX (Amsterdam exchange index, with spot value 443.83). The investments of the clearing members in the other assets are comparatively much smaller.

Figure \ref{f:heatmap_highest} shows the signed positions in the underlying assets of the twelve clearing members with the largest gross positions (left) and the signed positions of the clearing members in the nine most traded underlying assets (right), in a heatmap representation. In particular, we observe from the left panel that the biggest players in the CCP, namely the members labeled PB7, PB56, PB59 and PB50, have opposite sign positions in the main asset (the one with ticker FCE). The right panel shows that the dominant asset position in the CCP, i.e. the one in FCE, is shared (with opposite signs) between a significant number of clearing members.
Figure \ref{f:vols} shows the annualized volatilities $\kappa_i \times \sqrt{\frac{\nu_i}{\nu_i -2}} \times \sqrt{\frac{250}{3}}$ of the underlying assets (cf. \eqref{e:Si}). Most of these volatilities are comprised between 15\% and 40\%, with two assets, KBC and TMS, spiking over 60\% volatility. However, the clearing members are only very marginally invested in these two assets (their tickers do not even appear in the right panel of Figure \ref{f:gross_pos_udl}).
Figure \ref{f:log_risk_risk} shows the monetary risks (3d volatilities $\times$ absolute monetary positions) in the underlying assets of the ten 
clearing members with the largest gross positions.
From the right panel we see that the FCE and AEX assets (CAC40 index future FCE and Amsterdam exchange index AEX, two major indices) concentrate 
most of the risk of the clearing members. The comparison with Figure \ref{f:vols} shows that this is not an effect of the volatility of these assets, but of very large monetary positions of the clearing members.

\begin{figure}[htbp]
    \centering
    \includegraphics[width=0.45\textwidth]{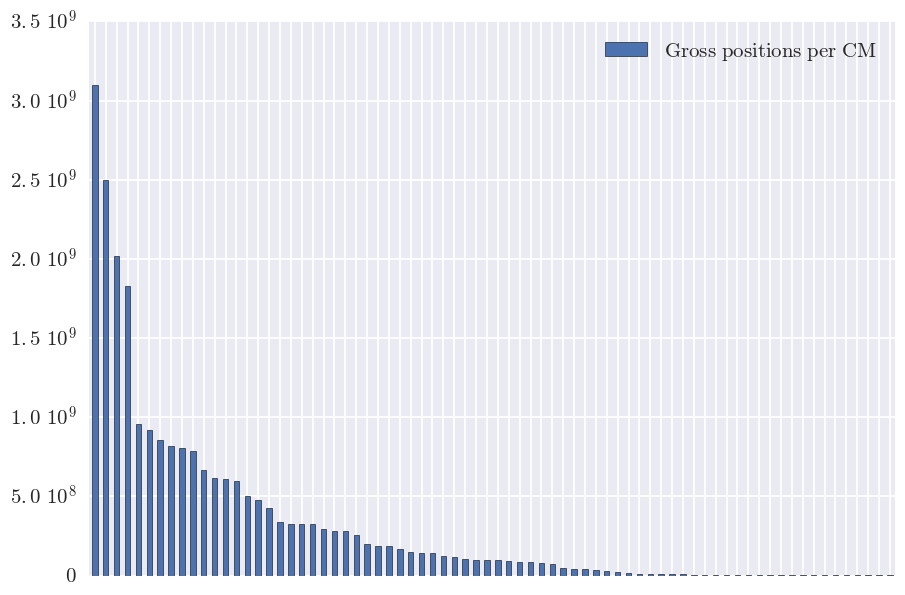}
 	\includegraphics[width=0.45\textwidth]{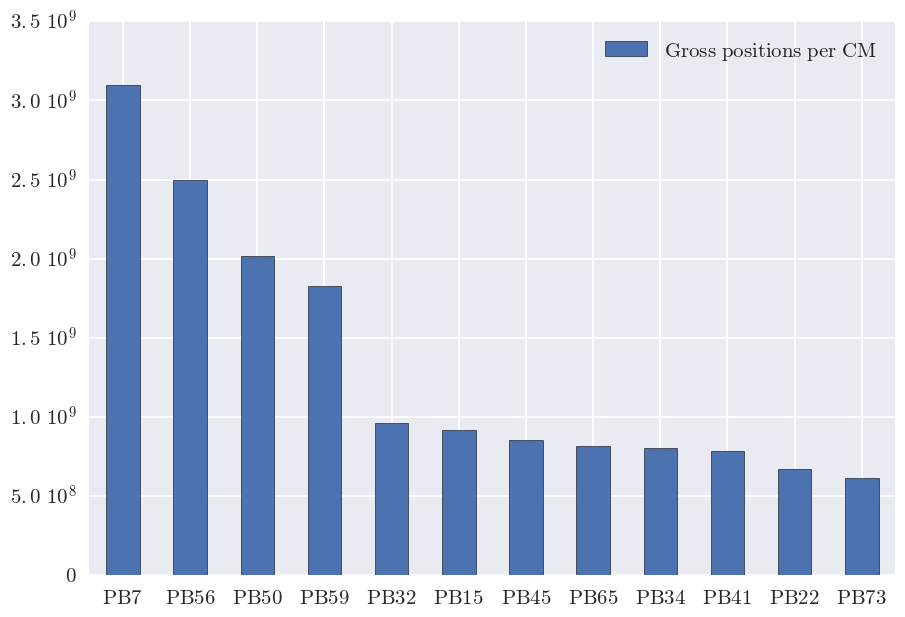}
    \caption{	Left: Gross positions per clearing member, ranked decreasing. 
    			Right: Zoom on the left part of the graph with member labels.}

    \label{f:gross_pos_cm}
\end{figure}

\begin{figure}[htbp]
    \centering
    
    \hspace{6pt} 
    
    \includegraphics[width=0.45\textwidth]{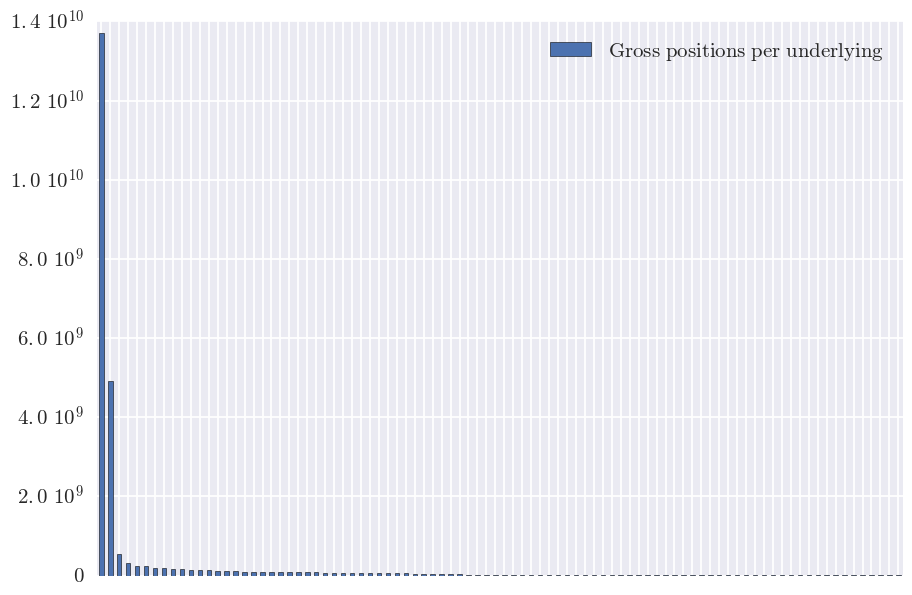}
 	\includegraphics[width=0.45\textwidth]{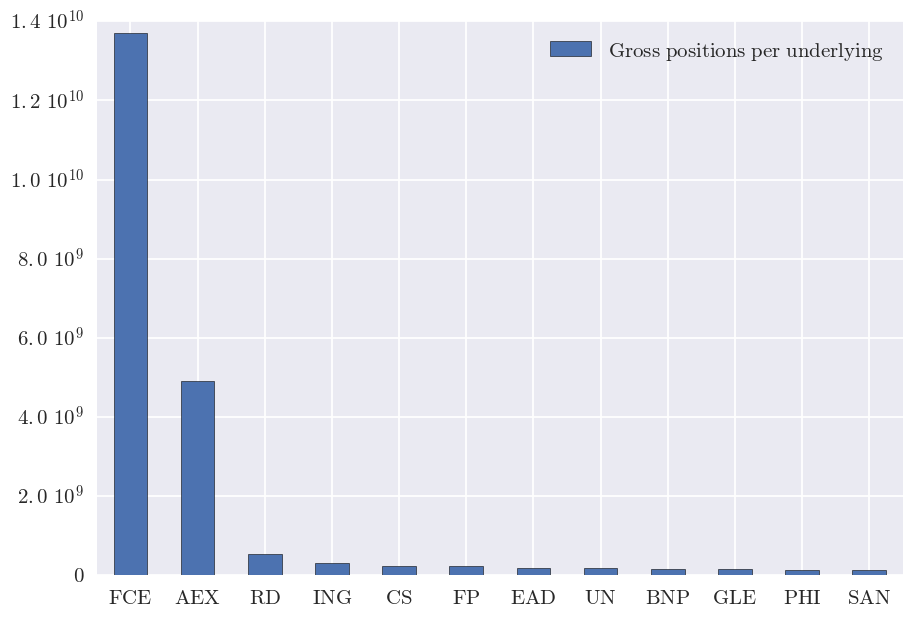}
    
    \includegraphics[width=0.45\textwidth]{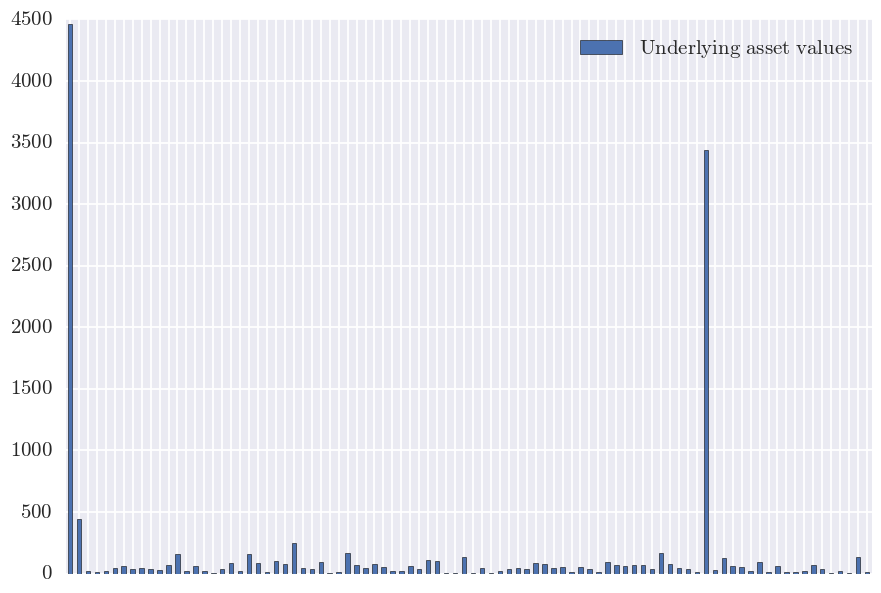} 
    \includegraphics[width=0.45\textwidth]{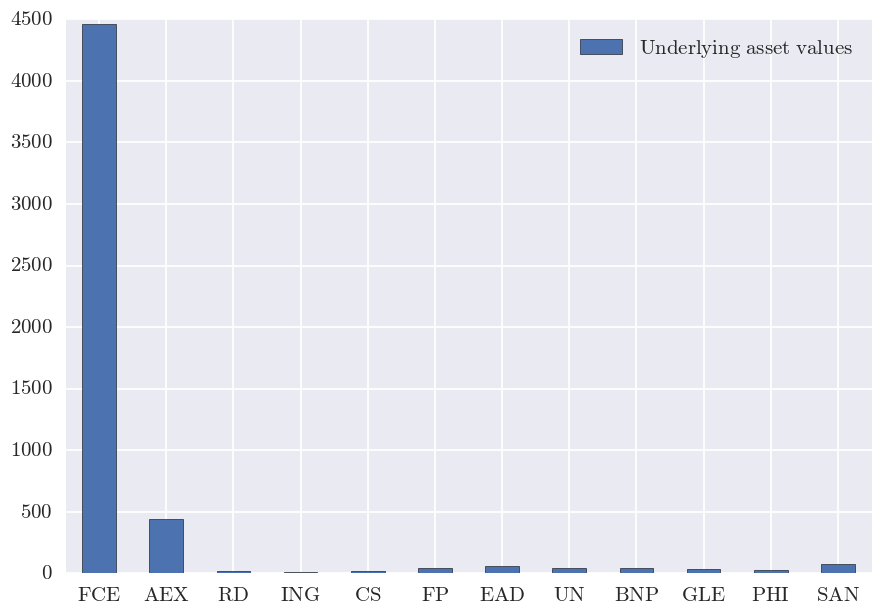}
    
    \caption{	Top: Gross positions per underlying, ranked decreasing (left) 
    			and zoom on the left part of the graph with tickers (right). 
    			Bottom: Spot values of the underlying assets, ranked as above (left) 
    			and zoom on the left part of the graph with tickers (right).}
    
    \label{f:gross_pos_udl}
\end{figure}

\begin{figure}[htbp]
    \centering
    \includegraphics[width=0.45\textwidth]{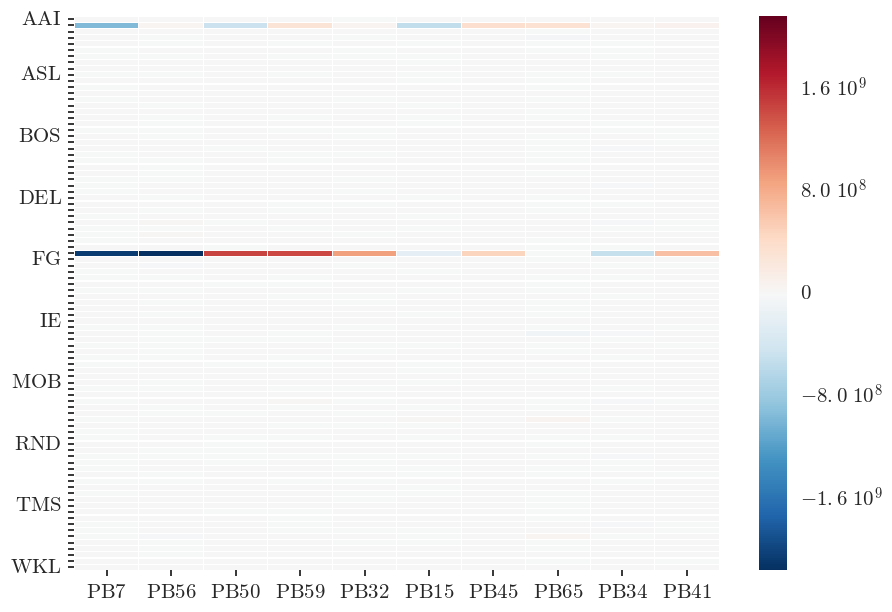}
    \includegraphics[width=0.45\textwidth]{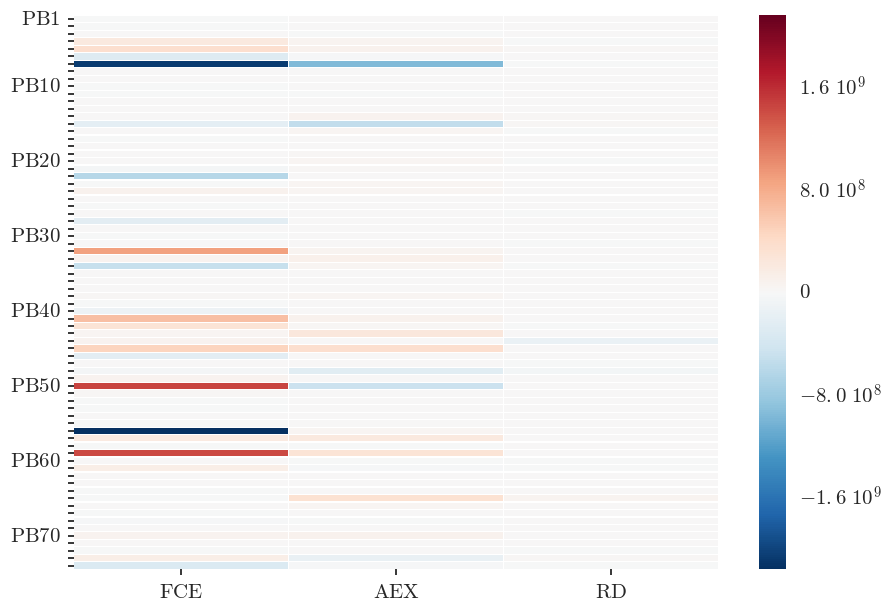}

    \caption{	Left: Positions in the underlying assets (one ticker out of ten displayed along the $y$ axis) 
    			of the ten clearing members with the largest gross positions, ranked by decreasing gross positions. 
    			Right: Positions of the clearing members (one label out of ten displayed along the $x$ axis) 
    			in the three most invested-in underlying assets, ranked by asset gross positions of the CCP.}

    \label{f:heatmap_highest}
\end{figure}

\begin{figure}[htbp]
	\centering
	\includegraphics[width=0.45\textwidth]{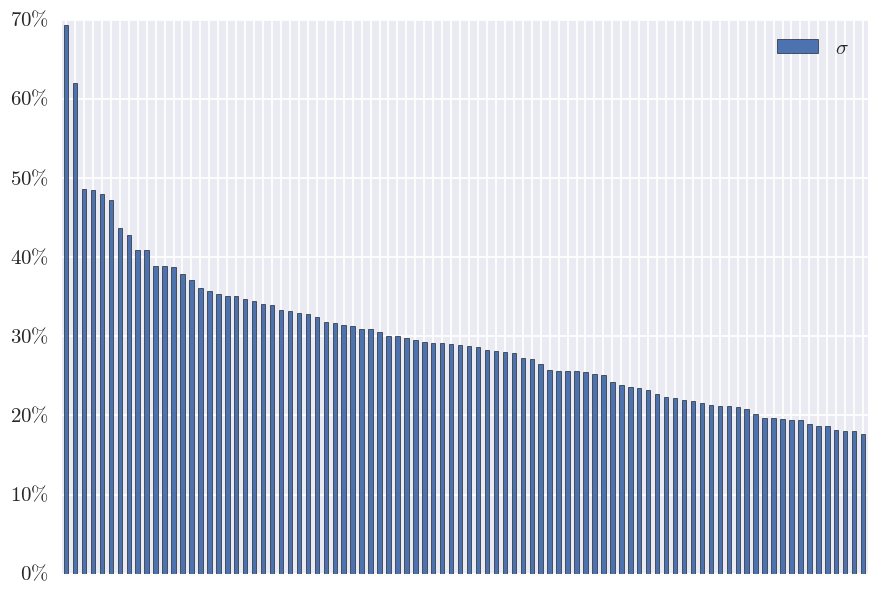}
	\includegraphics[width=0.45\textwidth]{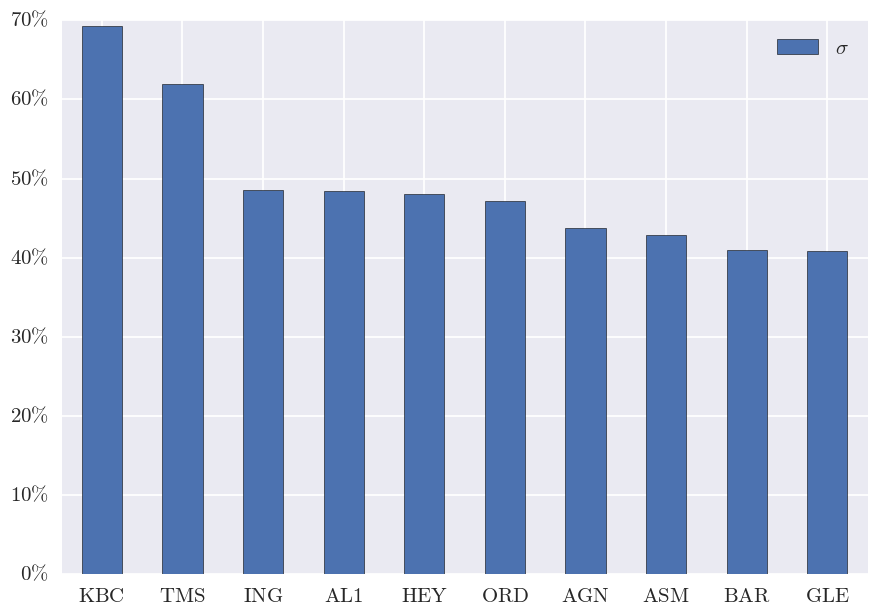}

	\caption{	Left: Underlying asset volatilities (ranked by decreasing order).
				Right: Zoom on the left part of the graph with tickers.}

	\label{f:vols}
\end{figure}


\begin{figure}[htbp]
    \centering
    \includegraphics[width=0.45\textwidth]{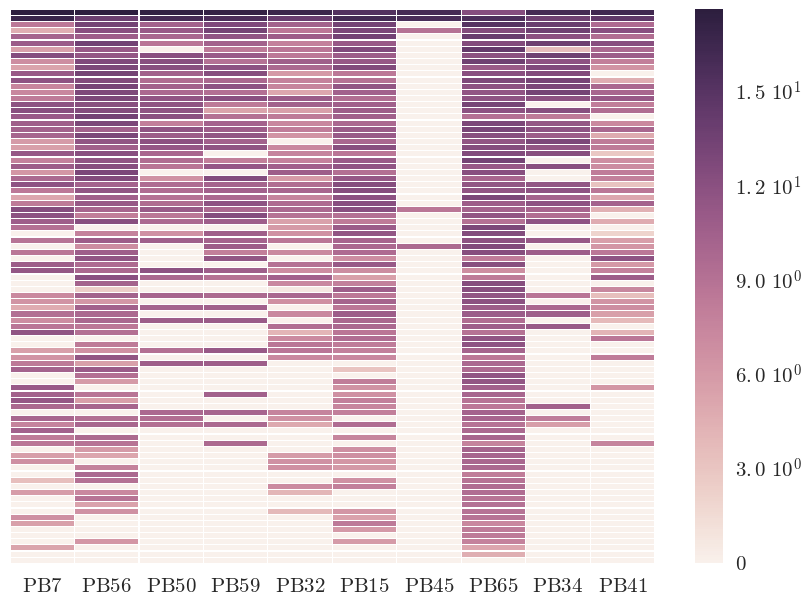}
    \includegraphics[width=0.45\textwidth]{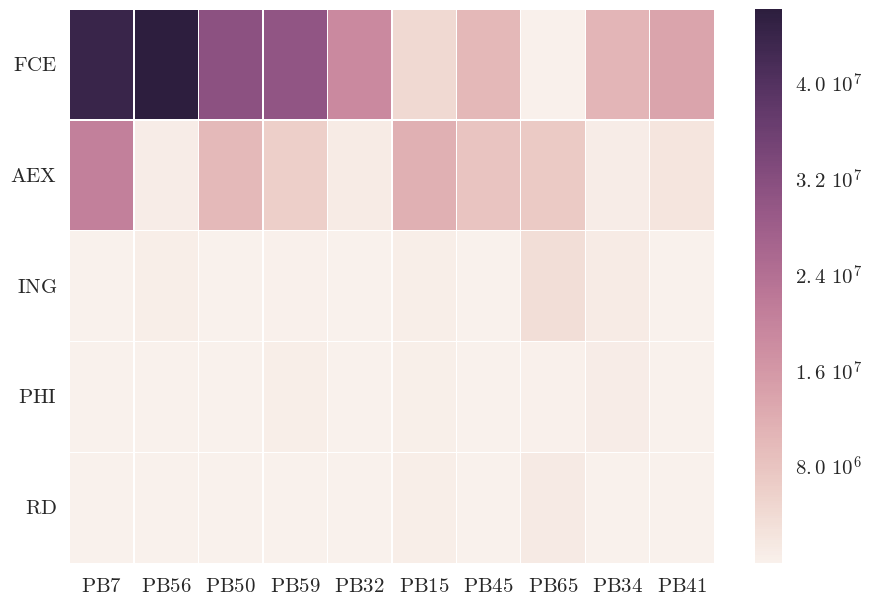}

    \caption{	Left: Log monetary risks in the underlying assets, ranked by decreasing risk order, 
    			of the ten clearing members with the largest gross positions.
    			Right: Monetary risks in the five most invested-in underlying assets of the ten 
    			clearing members with the largest gross positions.}

    \label{f:log_risk_risk}
\end{figure}



\end{appendix}
\newpage
\bibliographystyle{abbrvnat}
\bibliography{biblio}

\end{document}